%% file: D2DMEC-arxiv.tex
\documentclass[12pt,transaction, draftclsnofoot,onecolumn]{IEEEtran}

\usepackage{amsfonts}
\usepackage{amssymb}
\usepackage{hyperref}
\usepackage{graphicx}
\usepackage{subfigure}
\usepackage{enumerate}
\usepackage{amsmath}
\usepackage{color}
\usepackage{amsthm}
\usepackage{amsmath}
\usepackage{algorithm}
\usepackage[noend]{algpseudocode}

\input{macros.tex}

\ifodd 1
\newcommand{\congr}[1]{{\color{blue}#1}}
\else
\newcommand{\congr}[1]{#}
\fi

\ifodd 1
\newcommand{\congc}[1]{{\color{red}(Cong: #1)}}
\else
\newcommand{\congc}[1]{}
\fi

\begin{document}
%
\title{Designing Security-Aware Incentives for Computation Offloading via Device-to-Device Communication}

%
%
%

\author{Jie~Xu, Lixing~Chen, Kun~Liu, Cong~Shen
\thanks{J. Xu and L. Chen are with the Electrical and Computer Engineering Department, University of Miami, USA. K. Liu and C. Shen are with the Electronic Engineering and Information Science Department, University of Science and Technology of China.}
}

\maketitle

\begin{abstract}
Computation offloading via device-to-device (D2D) communication, or D2D offloading, has recently been proposed to enhance mobile computing performance by exploiting spare computing resources of nearby user devices. The success of D2D offloading relies on user participation in collaborative service provisioning, which incurs extra costs to users providing the service, thus mandating an incentive mechanism that can compensate for these costs. Although incentive mechanism design has been intensively studied in the literature, this paper considers a much more challenging yet less investigated problem in which selfish users are also facing interdependent security risks, such as infectious proximity-based attacks. Security cost is significantly different in nature from conventional service provisioning costs such as energy consumption, since security risks often depend on the collective behavior of all users. To this end, we build a novel mathematical framework by leveraging the combined power of game theory and epidemic theory to investigate the interplay between user incentives and interdependent security risks in D2D offloading, thereby enabling the design of security-aware incentive mechanisms. Our analysis discovers an interesting ``less is more'' phenomenon: although giving users more incentives promotes more participation, it may harm the network operator's utility. This is because too much participation may foster persistent security risks and as a result, the effective participation level does not improve. Our model and analysis shed new insights on the optimization of D2D offloading networks in the presence of interdependent security risks. Extensive simulations are carried out to verify our analytical conclusions.
\end{abstract}


%
\IEEEpeerreviewmaketitle

\section{Introduction}
More and more mobile applications nowadays demand resources (e.g. processer power, storage, and energy) that frequently exceed what a single mobile device can deliver. To meet the challenge of limited resources of individual mobile devices, mobile cloud computing (MCC) has been proposed and extensively studied in the literature (\cite{dinh2013survey} and references therein), which offloads computation tasks from mobile devices to the remote cloud for processing. To further reduce latency and enable precise location/context awareness, mobile edge computing (MEC)  \cite{mao2017mobile,chiang2016fog} has recently emerged as a new paradigm in which computation and storage capabilities are moved from the central cloud to the network edge devices such as base stations (BSs).

In both MCC and MEC, mobile users access computing services on the cloud or edge servers co-located with BSs through high-speed and ubiquitous wireless connection. However, the ever-growing number of mobile devices and volume of sensory data pose increasingly heavy traffic burden on the wireless network. Moreover, edge servers are limited in their computing power and storage capacity and hence are likely to be overloaded due to the increasing computation demand. To overcome these drawbacks, device-to-device (D2D) communication is exploited for offloading computation to nearby peer mobile devices, thereby fully unleashing the potential of the computation power on the mobile devices \cite{marinelli2009hyrax,shi2012serendipity,huerta2010virtual,li2014can}. This technique, known as D2D offloading, can not only increase the wireless network capacity but also alleviate computation burden from edge servers. However, the D2D architecture poses many challenges on computation offloading and two key challenges are incentives and security: because providing offloading service incurs extra costs (e.g. computation/transmission energy consumption) to the mobile users acting as servers, incentives must be devised to encourage selfish users to participate; because D2D offloading relies on ordinary mobile users whose security protection is much weaker than the operator network, D2D offloading is much more vulnerable to various types of attacks.

A commonly seen security risk in D2D offloading is proximity-based infectious attacks \cite{mtibaa2015friend,lu2014can,lu2016evolution}, which has been widely studied in mobile networks. In such attacks, mobile user equipments can get compromised by proximity-based mobile viruses when they are acting as D2D servers and consequently are unable to provide service to other users. Moreover, compromised users become new sources of attack when they communicate to other users via D2D in the future, thereby spreading the attack across the entire network. Although incentives and security are often treated as separate topics in the literature, they are indeed intricately intertwined in the D2D offloading system. On the one hand, the outcome of attack depends on users' collective decisions on the participation in D2D offloading and hence, the risk is interdependent among all users. Intuitively, more participation fosters faster and wider spread of the attack and hence may cause a greater damage to the overall system. On the other hand, the security risk shapes individual users' participation incentives. In view of a larger chance of being compromised, individual users may strategically reduce their participation levels, thereby degrading the performance of computation offloading. Although there is a huge literature on incentive mechanism design, systematic understanding of participation incentives under interdependent security risks and their impacts on the network performance is largely missing.

In this paper, we make the first attempt to investigate the interplay between incentives and interdependent security risks in D2D offloading  and design security-aware incentive mechanisms. To this end, we build a novel analytical framework by leveraging the combined power of game theory \cite{fudenberg1991game} and epidemic theory \cite{kermack1927contribution}.  Although we focus on D2D offloading networks in this paper, the techniques developed can be easily extended to other D2D systems. The main contributions of this paper are as follows:

(1) We focus on the utility maximization of the wireless network operator who provides both communication and computation services. The problem is formulated as a Stackelberg game in which the operator is the leader, who designs the incentive mechanism for D2D offloading, and the mobile users are the followers, who decide their D2D participation levels. Unlike conventional Stackelberg games, users not only perform best response to the operator's decision but also play another game among themselves due to the interdependent security risks.

(2) We characterize individual users' participation incentives under infectious proximity-based risks with tunable user foresightedness. Users' participation strategy is shown to have a threshold structure: a user is willing to provide wireless D2D offloading service if and only if the risk (i.e. the chance of serving a compromised requester) is low enough.

(3) We analyze the infection propagation dynamics when selfish users are strategically determining their participation levels. This analysis is in stark contrast with existing epidemic studies which assume non-strategic users obediently following prescribed and fixed actions. A key result is that there is a phase transition effect between persistent and non-persistent infection, which is substantially different in nature from the non-strategic user case.

(4) We discover an interesting ``Less is More'' phenomenon: although offering users a higher reward promotes participation, a too high participation level degrades the system performance and consequently reduces the operator's utility. This is because too much participation fosters \textit{persistent} infection and hence, the \textit{effective} participation level does not improve. Our result enables the optimal reward to be determined by solving a simple optimization problem.

(5) We perform extensive simulations to verify our analytical results for various mobility models, risk parameters and user heterogeneity levels.

The rest of this paper is organized as follows. Section II reviews related work. Section III presents the system model. Individual participation incentives are studied in Section IV. Section V investigates the epidemic dynamics. Section VI studies the optimal reward mechanism design problem. Simulations are presented in Section VII, followed by conclusions in Section VIII.

\section{Related Work}
The concept of offloading data and computation is used to address the inherent problems in mobile computing by using resource providers other than the mobile device itself to host the execution of mobile applications \cite{fernando2013mobile}. This paper focuses on D2D offloading which uses nearby mobile devices as resource providers via D2D communication \cite{marinelli2009hyrax,shi2012serendipity,chatzopoulos2016have}. For instance, ``Hyrax'' proposed in \cite{marinelli2009hyrax} explores the possibility of using a cluster of mobile phones as resource providers and shows the feasibility of such a mobile cloud. ``Serendipity'' \cite{shi2012serendipity} is another prominent work that proposes and implements a framework that uses nearby devices for distributed task computation. In general, computation offloading is concerned with what/when/how to offload a user's workload from its device to other resource providers (see \cite{huang2012dynamic,satyanarayanan2009case} and references therein). The problem studied in this paper is not contingent on any specific offloading technology. Rather, we design incentives in the presence of interdependent security risks and hence, our approach can be used in conjunction with any existing offloading technology.

D2D offloading benefits from the fact that mobile users in close proximity can establish direct wireless communication link over the licensed spectrum (\textit{inband}) or unlicensed spectrum (\textit{outband}) while bypassing the cellular infrastructure such as the BSs \cite{chang2017collaborative,chen2015computation}. This new communication paradigm can significantly improve the system throughput \cite{Min2011}, energy efficiency \cite{Feng2015}, fairness \cite{Poor2014} and overall QoS performance \cite{Asadi2014}. In practice, D2D communication has been implemented in industry products such as Qualcomm FlashLinQ \cite{FlashLinQ} and standardized by 3GPP \cite{Lin2014}. Enabled by D2D communication, D2D offloading can further alleviate computation burdens from the edge computing infrastructure \cite{li2014exploring}. A general consensus in the literature is that mobile users would have no incentive to participate in D2D service provisioning unless they receive satisfying rewards from the network operator \cite{li2015incentive}. However, despite the large body of work on interference management and resource allocation in D2D communication \cite{tanbourgi2014cooperative,zhang2013interference}, there are much fewer works to address the problem of providing incentives for users to provide D2D service. Our prior works \cite{xu2013token,mastronarde2016relay} design token-based incentive mechanism to promote wireless D2D relaying. A contract-based incentive mechanism is developed in \cite{zhang2015contract} to let users self-reveal their preferences. Market models are developed in \cite{li2015incentive} in which Stackelberg games and auctions are used to design incentive mechanisms for open markets and sealed markets, respectively. However, these existing works do not consider the interdependent security risks in the D2D interactions.


One of the greatest challenges for D2D offloading is the interdependent security risk among users. Although offloading in general faces many security and privacy issues, this paper focuses on proximity-based infectious attacks which rely on the cooperative interaction among users. Extensive research has shown that such attacks can be easily initiated in wireless mobile networks and cause significant damage \cite{mtibaa2015friend,lu2014can}. To model such attacks, we utilize the popular Susceptible-Infected-Susceptible (SIS) model \cite{kephart1991directed,wang2003epidemic,bailey1975mathematical,van2009virus} from the epidemics research community, which is a standard stochastic model for virus infections and widely-adopted to investigate computer virus/worm propagation.
Existing works in this regard can be divided into two categories. The first category adopts a mean-field approximation to study networks consisting of a large number of individuals \cite{kephart1991directed,bailey1975mathematical}. The second category tries to understand the influence of graph characteristics on epidemic spreading when users are interacting over a fixed topology \cite{wang2003epidemic}\cite{van2009virus}. Since mobile users are numerous and server-requester matching is short-lived and random due to user mobility and task arrival processes, the mean-field approach provides a good model approximation for the D2D offloading architecture. A classic result of these models is that there exists a critical effective spreading rate below which the epidemic dies out \cite{kephart1991directed,wang2003epidemic}. However, all these works study \textit{non-strategic} users who are following given and fixed actions. The present paper significantly departs from this strand of literature and studies \textit{strategic} agents who choose their participation levels to maximize their own utility.

\section{System Model}
\subsection{Network Model and Incentive Mechanism}
We consider a continuous time system and a wireless network in which mobile user equipments (UEs) generate computation tasks over time. We adopt a continuum population model for UEs to capture the large number of UEs in the network. The wireless network operator provides MEC service so that UEs can offload their data and tasks to the edge servers, which are often co-located with the BSs. When the edge server reaches its computation capacity or the wireless network is congested, the operator will try to employ D2D offloading, if possible, as a supplement in order to fulfill UEs' computation requests. In this case, the task data will be transmitted to nearby UEs that have spare computing resources via wireless D2D communication (e.g. Wifi-Direct \cite{wifidirect} or LTE-Direct \cite{ltedirect}), which is assumed to be error-free. To facilitate the exposition, we call the UE who requests the offloading service as the \textit{requester} and the UE who provides the service as the \textit{server}. The considered wireless network is \textit{dynamic} in two senses. First, UEs are moving in the network and hence the physical topology of the network is changing. We consider a generic mobility model that makes the users strongly mixed. Second, each UE can be a requester when it has demand and can also be a server when it is idle. As a result, the logical matchings of requesters and servers are also changing over time depending on both the physical topology and the demand arrival process. It is worth noting that there are a lot of concurrent D2D offloading instances going on at the same time in the network. Figure \ref{network} shows a snapshot of part of the network in which some UEs are requesters and some UEs are potential servers.

\begin{figure}
  \centering
  \includegraphics[width=2.5in]{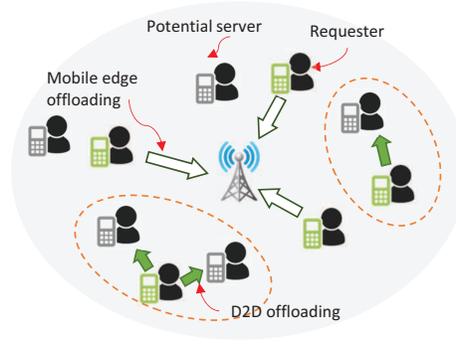}\\
  \caption{Snapshot of one of the cells in the network.}\label{network}
  \vspace{-10pt}
\end{figure}

For each task completed via D2D offloading, the operator obtains a benefit due to saved wireless bandwidth and computation cost. The expected value of this benefit is denoted by $b_0$. On the other hand, D2D offloading incurs an extra cost to the UE acting as the D2D server due to computation and transmission energy consumption and hence, selfish UEs are reluctant to provide the D2D service unless proper incentives are provided. Let the expected cost incurred to UE $i$ by completing one task be $c_i$, which differs across UEs. Although the realized cost depends on the specific computation task and the instantaneous wireless channel condition, for the UE's decision making purpose, we consider only the expected cost. In order to motivate participation in providing D2D offloading service, the operator offers reward to the participating UEs, in forms of free data or monetary payments, and how much reward a UE can receive depends on its participation level.

We use contracts as the incentive mechanism as in \cite{zhang2015contract}. Specifically, each UE $i$ makes a participation-reward bundle contract $(a^t_i, r(a^t_i))$ with the network operator at any time $t$ where $a^t_i$ is the  participation level chosen by UE $i$, and $r(a^t_i)$ is the unit time reward offered by the operator (e.g. monetary subsidy or increased data service rate). A contract $(a^t_i, r(a^t_i))$ requires that UE $i$ provides D2D offloading service with a rate up to $a^t_i$ tasks (from possibly different requesters) per unit time. The reward $r(a)$ is an increasing function of $a$. As a practical scheme, the operator adopts a throttled linear reward function
\begin{equation}
r(a) = \left\{
\begin{array}{ll}
r_0 a, &\forall a \leq R_{max}/r_0\\
R_{max}, &\forall a > R_{max}/r_0
\end{array}\right.
\end{equation}
where $r_0$ is the unit reward and $R_{max}$ is the maximum reward that the operator is willing to pay. Let $M\triangleq R_{max}/r_0$. Such a throttled scheme enables easy system implementation and similar schemes are widely adopted by operators in the real world\footnote{For example, ATT has a data reward program \cite{att} in which users earn extra data by downloading games and apps or shopping in participatory stores with a data transfer capping of 1,000MB per bill period.}. Nevertheless, our framework and analysis can also be applied to the reward scheme without throttling. In practice, contracts have a minimum duration. To simplify our analysis, we assume that the minimum contract duration is small enough so that UEs can effectively modify their contracts at any time.

\subsection{Participation Incentives in the Attack-Free Network}
The operator assigns computation tasks to UEs according to their chosen participation levels whenever D2D offloading is needed. Due to UE mobility and randomness in the task arrival, we model the D2D offloading request arrival assigned to UE $i$ as a Poisson process with rate $a^t_i$ (which is a result of UE decision). The unit time utility function of UE $i$ by choosing a participation level $a^t_i$ is thus defined as
\begin{align}\label{utility}
u_i(a_i^t) = v_i(r_0 a_i^t) - c_i a_i^t
\end{align}
where $v_i(\cdot)$ is UE $i$'s evaluation function of the reward, which differs across UEs. Clearly, selfish UEs have no incentives to participate at a level greater than $M$ if the reward is throttled and hence, we will focus on the range $a^t_i \in [0, M]$. We make the following assumptions on $v(\cdot)$.
\begin{assumption} (1) $v'_i(\cdot) > 0$, $v''(\cdot) < 0$. (2) $v(0) = 0$, $v_i(r_0 M) > c_i M$. (3) $v'_i(r_0 M) < c_i/r_0 < v'_i(0)$.
\end{assumption}
Part (1) is the standard diminishing return assumption, namely the evaluation function is increasing and concave in the received reward. Part (2) states that zero participation brings zero utility, namely $u_i(0) = 0$, and the maximum participation yields at least a non-negative utility, namely $u_i(M) > 0$. Part (3) is equivalent to $\arg\max_a u_i(a) \in (0, M)$, namely the optimal participation level lies in $(0, M)$. We denote $a^\text{AF}_i(r_0) \triangleq \arg\max_a u_i(a)$ as the optimal participation level in the attack free (AF) network, which can be easily computed.  We  make a further assumption on $a^\text{AF}_i(r_0)$ as follows.

\begin{assumption}
$a^\text{AF}_i(r_0)$ is increasing in $r_0$.
\end{assumption}
Assumption 2 states that increasing reward provides UEs with more incentives to participate as a D2D server. This is a natural assumption and can be easily satisfied by many evaluation functions. For instance, if $v_i(x) = \sqrt{x}$, then $a^\text{AF}_i(r_0) = \frac{r_0}{4 c_i^2}$ which is increasing in $r_0$.

\subsection{Attack Model}
Participating as D2D servers exposes UEs to proximity-based infectious risks since the task data is sent directly from peer UEs via D2D communication rather than the trusted operator-owned BSs. We consider an attack whose purpose is to make the D2D offloading service unusable. Therefore compromised UEs will not be able to provide D2D offloading service. Moreover, they may become new sources of attack when they request offloading services from normal UEs in the future.

To model this attack, we adopt the popular Susceptible-Infected-Susceptible (SIS) epidemic model. At any time $t$, each UE $i$ is in one of the two states $s^t_i \in \{\text{(S)usceptible, (I)nfected}\}$, which indicates whether the UE is normal or compromised. The UE state is private information and unknown to either the operator or other UEs in the network. Otherwise, compromised UEs can be easily excluded. If a normal UE provides D2D offloading service to a compromised UE, then the normal UE gets infected by the virus with a probability $\beta\in[0,1]$. We assume that a compromised server does not infect (or with a negligible probability infects) a normal requester because detection is much more effective on the requester side due to the significantly smaller data size of the returned computation result \cite{mao2017mobile}.

Once a UE $i$ is compromised, it has to go through a recovery process. During this process, UE $i$ cannot provide any D2D offloading service to other UEs. Moreover, UE $i$ suffers a recovery cost $q_i$ per unit time. However, UE $i$ can still make offloading requests so that the virus can be propagated to other UEs. Nevertheless, our framework can be easily extended to handle the case where some compromised UEs are completely down so that they cannot make requests before recovery. We assume that the recovery process duration is exponentially distributed with mean $1/\delta$. Recovered UEs re-enter the normal state and may be compromised again in the future. The UE state transition is illustrated in Figure \ref{transition}.

The parameters $\beta$ and $\delta$ describe the intrinsic risk level of the wireless network and we define the \textit{effective infection rate} as $\tau \triangleq \beta/\delta$. We will treat these risk parameters as fixed for most part of this paper. In this way, we focus on the incentive mechanism design for a network in a given risk environment. In the later part of this paper (Section VI), we will discuss how the incentive mechanism and security technologies can be jointly designed.

\begin{figure}
  \centering
  \includegraphics[width=3.5in]{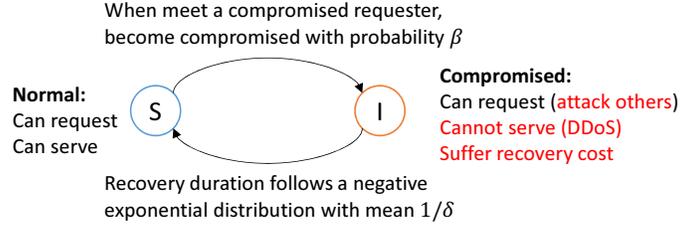}\\
  \caption{UE state transition.}\label{transition}
  \vspace{-10pt}
\end{figure}

\subsection{Problem Formulation as a Stackelberg Game}
The objective of the operator is to design the unit reward $r_0$ in order to maximize its own utility. The operator's utility is defined as:
\begin{align}
u_0(r_0) =  \mathbb{E}_{t,i} [(b_0 - r_0)\textbf{1}\{s^t_i = S\} a^t_i ]
\end{align}
where the expectation is taken over the UE attributes (i.e. the distribution of the evaluation function $v_i(\cdot)$ and D2D offloading cost $c_i$) and time. Clearly we must have $r_0 < b_0$. Otherwise, the operator will not adopt D2D offloading. The incentive mechanism design problem can be formulated as a Stackelberg game among a leader and an infinite number of followers (due to the continuum population model). The operator plays as the leader, which moves first and determines the reward mechanism $r_0$. The UEs play as the followers, which move next and choose their participation levels. In the attack-free network, the Stackelberg game can be represented by the following two-level maximization problem
\begin{align}
\max_{r_0} & ~~~(b_0  - r_0 )\mathbb{E}_i [ a^\text{AF}_i(r_0)]\\
\textit{s.t.} & ~~~a^\text{AF}_i(r_0) = \arg\max_a u_i(a|r_0), \forall i
\end{align}
where $\mathbb{E}_{t,i} [\textbf{1}\{s^t_i = S\} a^t_i ]$ is replaced with $\mathbb{E}_i [ a^\text{AF}_i(r_0)]$ since there is no attack and hence the UEs are never compromised. This problem is not difficult to solve as $a^\text{AF}_i(r_0)$ can be easily computed.

The presence of infectious attacks changes both the operator's objective function and UEs' participation incentives. First, because UEs may get compromised and consequently are not able to provide D2D computing service sometimes, the utility that the operator can reap from D2D offloading depends on not only the UEs' participation decisions but also the network security state (i.e. the fraction of normal UEs). Therefore, the operator's utility becomes $(b_0  - r_0)\mathbb{E}_i [ \textbf{1}\{s^t_i = S\} a^*_i(r_0)]$ and we call $\mathbb{E}_i [ \textbf{1}\{s^t_i = S\} a^*_i(r_0)]$ the \textit{effective} participation level of D2D offloading. Second, UE $i$'s incentive to participate in D2D offloading will be determined by a utility function $U_i(a_i)$ that incorporates the potential future infection (which will be characterized in the next section). What much complicates the problem is the fact that UEs face interdependent security risks -- the security posture of the whole network depends on not only the participation action of UE $i$ itself but also all other UEs' participation since the infection is propagated network-wide. Therefore, the utility function $U_i(a_i)$ should indeed be a function of all UEs' actions and hence, we denote it by $U_i(a_i, \a_{-i})$ where $\a_{-i}$, by convention, is the action profile of UEs except $i$. Formally, the Stackelberg game for the security-aware incentive mechanism design problem is:
\begin{align}
\max_{r_0} & ~~~(b_0  - r_0)\mathbb{E}_i [ \textbf{1}\{s^t_i = S\} a^*_i(r_0)]\\
\textit{s.t.} & ~~~a^*_i(r_0) = \arg\max_{a} U_i(a_i, \a_{-i}|r_0), \forall i
\end{align}
\textbf{Remark}: The formulated Stackelberg game is significantly different from conventional Stackelberg games. In our game, the followers (UEs) not only perform best response to the leader (i.e. the operator)'s decision, but also play a different yet coupled game among themselves due to the interdependent security risk. To solve this Stackelberg game, we will use backward induction to firstly investigate the participation incentives of UEs under the interdependent security risk and the resulting \textit{effective} participation level $\mathbb{E}_i [ \textbf{1}\{s^t_i = S\} a^*_i(r_0)]$, and then design the optimal reward mechanism.

\section{Individual Participation Incentives}
\subsection{Foresighted Utility}
If a UE is myopic and only cares about the instantaneous utility, then the UE will simply choose a participation level $a^*_i =  \arg\max_{a} [v_i(r_0 a) - c_i a]$ which maximizes the instantaneous utility when it is in the normal state. In such cases, $a^*_i = a^{AF}$. However, since the infectious attack may cause potential future utility degradation, the UE will instead be foresighted and care about the \textit{foresighted utility} \cite{mailath2006repeated}. Since the infection risk depends on the probability that a server UE meets a compromised requester UE, the fraction of compromised UE in the system at any time $t$, denoted by $\theta^t \in [0,1]$, plays a crucial role in computing the foresighted utility. Assume that requester-server matching for D2D offloading is uniformly random, then the probability that server UE $i$ meets a compromised requester UE is exactly $\theta^t$. We thus call $\theta^t$ the \textit{system compromise state} at time $t$. Note that $\theta^t$ is an outcome of all UE's participation decisions. The foresighted utility is defined as follows.

\begin{definition}
(\textbf{foresighted utility}) Given the system compromise state $\theta$, the foresighted utility of UE $i$ with discount rate $\rho$ when it takes a participation action $a_i$ is defined as
\begin{align}\label{def}
U(a_i, \theta) = \rho \int_{t=0}^\infty \big(\underbrace{\int_{\tau=0}^t e^{-\rho\tau} u_i(a_i)d\tau}_{\substack{\text{discounted sum utility} \\ \text{before being infected}}} + \underbrace{e^{-\rho t} U_I}_{\substack{\text{discounted continuation utility} \\ \text{after being infected}}}\big) \underbrace{\theta \beta a_i e^{-\theta \beta a_i t}}_{\substack{\text{exponential distribution}\\\text{due to Poisson arrival}}} dt
\end{align}
where
\begin{align}\label{def2}
U_I = \int_{t=0}^\infty \left(\int_{\tau=0}^t e^{-\rho \tau} (-q_i) d\tau + \rho^{-1} e^{-\rho t} U(a_i, \theta)\right) \delta e^{-\delta t} dt
\end{align}
\end{definition}

We explain the definition of foresighted utility below:

(i) Suppose that UE $i$ is compromised at time $t_0 + t$ where $t_0$ is the time when UE $i$ makes the current contract, then $\int_{\tau=0}^t e^{-\rho \tau} u_i(a_i) d\tau$ is the discounted sum utility that the UE can receive during the period $[t_0, t_0 + t]$ before it is compromised. The term $e^{-\rho \tau} \leq 1$ represents the discounting mechanism which decreases with $\tau$. A larger discount rate $\rho$ means that the discounting is greater. Using the exponential function $e^{-\rho \tau}$ to model discounting is the standard way for continuous time systems \cite{doya2000reinforcement,mailath2006repeated}.

(ii) $U_I$ is the utility that UE $i$ receives once it gets compromised. During the recovery phase, UE $i$ suffers a cost $-q_i$ and once it is recovered, it receives again the foresighted utility $U(a_i, \theta)$ which is discounted by $e^{-\rho t}$ where $t$ the duration of the recovery phase. Here we assumed \textit{bounded rationality}: when computing the foresighted utility, the UE believes that it will choose the \textit{same} participation level as before once it is recovered since it expects the system compromise state to stay the same in the future. In the steady state, this belief will in fact be correct.

(iii) $\theta \beta a_i e^{-\theta\beta a_i t}$ is the probability density function of being compromised at time $t$. The D2D request arrival is a Poisson process with rate $a_i$ according to the committed participation level. With probability $\theta$  UE $i$ meets a compromised requester UE and further with probability $\beta$ UE $i$ is compromised. As a basic property of the Poisson process, the infection process is also Poisson with arrival rate $\theta \beta a_i$.

(iv) The constant $\rho$ at the beginning of the right-hand side equation is a normalization term, which is due to $\int_{t=0}^\infty e^{-\rho t} dt = \rho^{-1}$.

\begin{figure}
  \centering
  \includegraphics[width=3.2in]{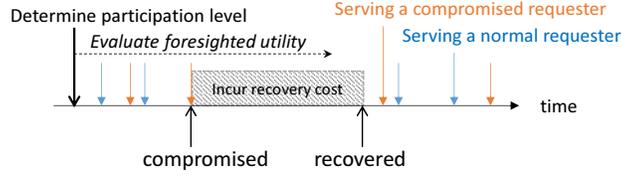}\\
  \caption{Illustration of foresightedness in participation level determination.}\label{foresighted}
  \vspace{-10pt}
\end{figure}

Figure \ref{foresighted} illustrates the role of foresightedness plays in determining the participation level of a UE. By solving \eqref{def} and \eqref{def2}, the foresighted utility can be simplified to
\begin{align} \label{FU}
U(a_i, \theta) = \frac{(\rho + \delta)u_i(a_i) - \beta\theta a_i q}{\rho + \delta + \beta \theta a_i}
\end{align}
\textbf{Remark}: The above form of foresighted utility generalizes the instantaneous utility and the time-average long-term utility, and can capture a wide spectrum of UE behaviors by tuning the discount rate $\rho$. When $\rho \to \infty$, the UE cares about only the instantaneous utility since future utility is infinitely discounted. In this case, $U(a_i, \theta)$ reduces to the myopic utility $u(a_i)$. When $\rho \to 0$, the UE does not discount at all and hence, $U(a_i, \theta)$ becomes the time-average utility $\frac{\delta u(a_i) - \beta\theta a_i q}{\delta + \beta \theta a_i}$, which is the same result by performing the stationary analysis of a continuous-time two-state Markov chain.

\subsection{Individual Optimal Participation Level}
In this subsection, we study the optimal participation level that UE $i$ will choose to maximize its foresighted utility.

\begin{proposition}
If the system compromise state $\theta \geq \frac{(r_0 v'_i(0) - c_i)(\rho + \delta)}{q_i \beta} \triangleq \bar\theta_i$, then the optimal participation level is $a^*_i(\theta) = 0$. Otherwise, the optimal participation level $a^*_i(\theta)$ is the unique solution of $u'_i(a_i)(\rho + \delta \beta\theta a_i) - u_i(a_i)\beta \theta - \beta\theta q_i = 0$, which increases as $\theta$ decreases.
\end{proposition}
\begin{proof}
We omit the UE index $i$ in the subscript of $v_i(\cdot)$, $c_i$ and $q_i$ for brevity. Given the foresighted utility in \eqref{FU}, determining the optimal participation level boils down to investigating the first-order condition of \eqref{FU}. The derivative of $U(a_i, \theta)$ is
\begin{align}
U'(a_i, \theta) = (\rho + \delta)\frac{u'(a_i)(\rho + \delta + \beta \theta a_i) - u(a_i) \beta\theta - \beta\theta q}{(\rho+\delta+\beta \theta a_i)^2}
\end{align}
For brevity, let $f(a_i) = u'(a_i)(\rho + \delta + \beta \theta a_i) - u(a_i) \beta\theta - \beta\theta q$, which has the same sign of $U'(a_i, \theta)$. First, we have
\begin{align}
f'(a_i) = u''(a_i)(\rho + \delta + \beta\theta a_i) = r^2_0 v''(r_0 a_i)(\rho+\delta + \beta\theta a_i) < 0
\end{align}
Therefore, $f(a_i)$ is monotonically decreasing. Next, we investigate the signs of $f(M)$ and $f(0)$.
\begin{align}
f(M) &= u'(M)(\rho+\delta+\beta\theta M) - u(M)\beta\theta - \beta\theta q< 0
\end{align}
The inequality is because $u'(M) < 0$ and $u(M) > 0$ according to Assumption 1(2). Also,
\begin{align}
f(0) = u'(0)(\rho + \delta) - u(0)\beta\theta - \beta \theta q = (r_0 v'(0) - c)(\rho + \delta) -  \beta \theta q
\end{align}
Therefore, if $\theta < \frac{(r_0 v'(0) - c)(\rho + \delta)}{q \beta}$, then $f(0) > 0$. Otherwise, $f(0) \leq 0$. This means that if $\theta \geq \frac{(r_0 v'(0) - c)(\rho + \delta)}{q \beta}$, then the optimal $a^* = 0$ and otherwise, there exists an optimal participation level $a^* > 0$, which is the unique solution of
\begin{align}\label{aSol0}
u'(a_i)(\rho + \delta \beta\theta a_i) - u(a_i)\beta \theta - \beta\theta q= 0
\end{align}
To investigate the monotonicity of $a_i^*$ with $\theta$, we rewrite the above equation as follows
\begin{align}\label{aSol}
\frac{u(a_i) + q}{u'(a_i)} - a_i = \frac{\rho+\delta}{\beta\theta}
\end{align}
Notice that $u(a_i)$ does not have $\rho$, $\delta$, $\beta$ or $\theta$ in its expression according to \eqref{utility}. The first-order derivative of the left-hand side function of $a_i$ is
\begin{align}\label{uMono}
\frac{u'(a_i)u'(a_i) - (u(a_i) + q) u''(a_i)}{(u'(a_i))^2} -1
= \frac{-(u(a_i) + q)u''(a_i)}{(u'(a_i))^2} > 0
\end{align}
The last inequality is because $u(a) > 0, \forall a\in [0,M]$ and $u''(a) = r^2_0 v''(r_0 a_i)< 0$. Therefore the left-hand side of \eqref{aSol} is monotonically increasing in $a_i$. Since the right-hand-side is decreasing in $\theta$, $a_i$ decreases with the increase of $\theta$.
\end{proof}

Proposition 1 can be intuitively understood. A larger system compromise state $\theta$ implies a higher risk of getting compromised via D2D offloading and hence, UE $i$ has smaller participation incentives. In particular, if the system compromise state exceeds a threshold, then UE $i$ will refrain from participating in the D2D offloading. It is also evident that the operator can provide UE with more incentives to participate by increasing the reward $r_0$.

\begin{proposition}
If $\theta < \bar\theta_i$, then $a^*_i(\theta)$ is increasing in $\rho, \delta$ and decreasing in $\beta$.
\end{proposition}
\begin{proof}
These claims are direct results of the monotonicity of the left-hand side of \eqref{aSol}.
\end{proof}
Proposition 2 states that a higher attack success probability (larger $\beta$) and a slower recovery speed (smaller $\delta$) both decrease UE $i$'s incentives to participate (smaller $\alpha^*$). Importantly, Proposition 2 also reveals the impact of UE foresightedness on the participation incentives: being more foresighted (smaller $\rho$) decreases the UE's participation incentives (smaller $\alpha^*$) because the UE cares more about the potential utility degradation caused by the attack. Figure \ref{UE participation} illustrates the individual optimal participation level for various $\theta$ and foresightedness $\rho$ via numerical results. Note that for $\theta = 1$, it is still possible that a UE chooses a positive participation level since even if it is always interacting with a compromised UE, the attack success probability $\beta$ is not 1.

\begin{figure}
  \centering
  \includegraphics[width=2.5in]{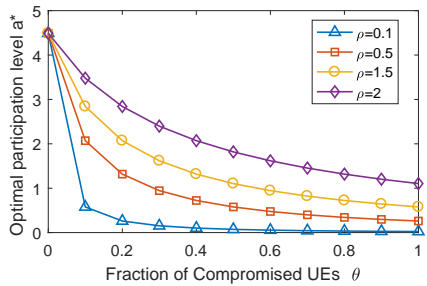}\\
  \caption{Illustration of individual optimal participation level.}\label{UE participation}
  \vspace{-10pt}
\end{figure}

\section{Interdependent Security Risks}
In this section, we study how the infection propagates in the system and the convergence of the system compromise state. Epidemic processes have been well investigated in the literature for different systems. Most of this literature assumes that users are obediently following a prescribed behavior rule. However, selfish UEs in the considered problem determine their participation levels to maximize their own foresighted utilities, thereby leading to significantly different results than conventional epidemic conclusions. To show this difference, we will first review the classic results of infection propagation in the context of the considered D2D offloading network. Then we will study the infection propagation processes for selfish UEs in two scenarios. In the first scenario, UEs observe the system compromise state at any time and hence make adaptive decisions according to the observed state. In the second scenario, UEs do not observe the system compromise state and hence have to conjecture this state based on the equilibrium analysis of a D2D offloading participation game.

\subsection{Attack under Given Participation Actions}
The evolution of the system compromise state $\theta^t$ depends on the strategy adopted by the UEs, the attack success rate $\beta$, the recovery rate $\delta$ as well as the initial state of the system $\theta^0$. It is obvious that if the system starts with an initial state $\theta^0 = 0$ (i.e. without initial attacks), then the system will remain in the state of zero compromise no matter what strategies are adopted by the UEs. Therefore, we will focus on the non-trivial case that the initial state $\theta^0 > 0$.

The network is said to be in the steady state if the system compromise state $\theta^t$ becomes time-invariant, denoted by $\theta^\infty$. The steady state reflects how the infectious attack evolves in the long-run. Existing works suggest that there exists a critical threshold $\tau_c$ for the effective infection rate $\tau$ such that if $\tau < \tau_c$, then the infectious attack extinguishes on its own even without external interventions, namely $\theta^\infty = 0$; otherwise, there is a positive fraction of compromised UEs, namely $\theta^\infty > 0$. This result is formally presented in our problem as follows.

\begin{proposition}
Assume that all UEs adopt the same fixed participation level $a$. Then there exists a critical effective infection rate $\tau_c = \frac{1}{a}$, such that if $\tau \leq \tau_c$, $\theta^\infty = 0$; otherwise, $\theta^\infty = 1 - \frac{\delta}{\beta a}$.
\end{proposition}
\begin{proof}
Consider the compromise state dynamics given a symmetric strategy $a$. For any $\theta$, the change in $\theta$ in a small interval $dt$ is
\begin{align}
d\theta = -\theta \delta dt + (1-\theta) \theta \beta a dt = \theta((1-\theta)\beta a - \delta)dt
\end{align}
Clearly, if $\tau > \frac{1}{a} \triangleq \tau_c$, then for any $\theta > 1 - \frac{\delta}{\beta a} \triangleq \theta^*$, $d\theta < 0$ and for $\theta < \theta^*$, $d\theta > 0$.  Therefore, the dynamic system must converge to $\theta^*$. If $\tau \leq \tau_c$, then for any $\theta > 0$, $d\theta < 0$. This means that the dynamic system converges to $0$.
\end{proof}
Notice that if $\beta < \delta$, then for all values of $a$, we must have $\tau \leq \tau_c$ and hence, the infection risk always extinguishes in the long-run.

\subsection{Attack with Strategic UEs and Observable System Compromise State}
In the considered D2D offloading system, UEs strategically determine their participation levels and hence, the infection propagation dynamics may be significantly different. In this subsection, we investigate the system dynamics assuming that UEs can observe the system compromise state $\theta^t$ at any time. In this case,  UEs adaptively change their participation levels by revising their contracts with the operator according to the observed system compromise state.

For the analysis simplicity, we assume that all UEs have the same evaluation function $v(\cdot)$, service provision cost $c$ and recovery cost $q$. We will investigate the heterogeneous case in the next subsection and in simulations. The system dynamics thus can be characterized by the following equation,
\begin{align}\label{systemdyn}
d \theta^t = -\theta^t \delta dt + (1-\theta^t) \beta\theta^t a^*(\theta^t) dt
\end{align}
where $a^*(\theta^t)$ is the best-response participation level given the current system compromise state $\theta^t$ according to our analysis in the previous section. In the above system dynamics equation, the first term $\theta^t \delta dt$ is the population mass of compromised UEs that are recovered in a small time interval $dt$. The second term $(1-\theta^t)\beta\theta^t a^*(\theta^t) dt$ is the population mass of normal UEs that are compromised in a small time interval $dt$. The following proposition characterizes the convergence of the system dynamics.

\begin{proposition}
There exists a critical effective infection rate $\tau_c = \frac{1}{a^{\text{AF}}}$ such that if $\tau < \tau_c$, then the system compromise state converges to $\theta^\infty = 0$; otherwise, $\theta^\infty = \theta^\dagger$ where $\theta^\dagger > 0$ is the unique solution of $(1-\theta^\dagger)a^*(\theta^\dagger) = 1/\tau$.
\end{proposition}
\begin{proof}
The system dynamics can be rewritten as
\begin{align}
d \theta^t = \theta^t[(1-\theta^t)\beta a^*(\theta^t) -\delta] dt
\end{align}
We are interested in the sign of $d\theta^t$ for different values of $\theta^t$. Since $\theta^t \geq 0$, what matters is the sign of $f(\theta^t) \triangleq (1-\theta^t)\beta a^*(\theta^t) -\delta$. For $\theta^t \geq \bar{\theta}$, $a^*(\theta^t) = 0$ and hence, $f(\theta^t) = -\delta < 0$. For $\theta^t < \bar{\theta}$, $f(\theta^t)$ is decreasing in $\theta^t$ because $a^*(\theta^t)$ is decreasing in $\theta^t$ according to Proposition 1. Now consider the sign of $f(0) = \beta a^*(0) - \delta$. According to \eqref{aSol0}, $a^*(0)$ is the solution to $u'(a) = 0$, which is the same as the optimal participation level in the attack-free case. Specifically, $a^*(0) = a^{\text{AF}}$. If $\beta a^\text{AF} - \delta < 0$, then $f(\theta^t) < 0$ for all $\theta^t$. Therefore, the system converges to $\theta^\infty = 0$. If $\beta a^\text{AF} - \delta \geq 0$, then there exists a unique point $\theta^\dagger \in [0, \bar\theta)$ such that for $\theta^t > \theta^\dagger$, $f(\theta^t) < 0$ and for $\theta^t < \theta^\dagger$, $f(\theta^t) > 0$. This means that the system compromise state will converge to $\theta^t$. Moreover, $\theta^t$ is the solution of $(1-\theta)\beta a^*(\theta) -\delta = 0$.
\end{proof}

Proposition 4 shows that when UEs are selfish and strategically deciding their participation levels, the infectious attack propagation also has a thresholding effect with regard to the effective infection rate. However, this thresholding effect is significantly different in nature from when UEs are obeying prescribed participation actions. In the non-strategic case, the effective infection rate threshold is a function of the given participation level. In the strategic case, the threshold is a constant. In particular, the constant is exactly the critical threshold when UEs follow the individually optimal action $a^\text{AF}$ in the attack-free network (see Figure \ref{critical} for an illustration).

\begin{figure}
  \centering
  \includegraphics[width=2.5in]{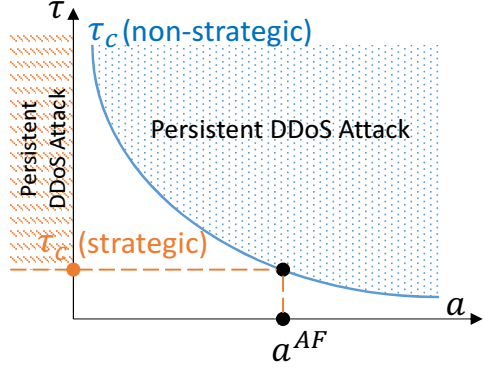}\\
  \caption{Critical effective infection rates in strategic and non-strategic cases.}\label{critical}
  \vspace{-10pt}
\end{figure}

Let us understand this thresholding effect better. As proved in Proposition 1, the individual optimal participation level under infectious attack risks is always lower than that in the attack-free network. Therefore, although the participation level is adapting over time depending on the system compromise state, it will never be higher than $a^\text{AF}$. As a result of Proposition 3, if the effective infection rate is less than $1/a^\text{AF}$, the attack will extinguish on its own. When the effective infection rate becomes sufficiently large (i.e. $\tau > 1/a^\text{AF}$), individual UEs always have incentives to choose a participation level greater than the threshold participation level that eliminates infection. This is because unilaterally increasing the participation level does not change the network-wide epidemic dynamics since each individual UE is infinitesimal in the continuum population model yet increases the benefit that the individual UE can obtain. Therefore infectious attacks persist.

Proposition 5 gives bounds on the convergence speed.
\begin{proposition}
If $\tau > \frac{1}{a^\text{AF}}$, then starting from any initial system compromise state $\theta^0 > \theta^\dagger$ (or $\theta^0 < \theta^\dagger)$, $\forall \epsilon > 0$, let $T(\theta^0, \epsilon)$ be the time at which $\theta^T$ decreases to $\theta^\dagger + \epsilon \triangleq \theta_\epsilon$ (or increases to $\theta^\dagger - \epsilon \triangleq \theta_\epsilon$), we have
\begin{align}
\frac{\ln \theta^0 -\ln \theta_\epsilon}{(1-\theta^0)\beta a^*(\theta^0) - \delta} < T(\theta^0, \epsilon) < \frac{\ln \theta^0 - \ln \theta_\epsilon}{(1-\theta_\epsilon)\beta a^*(\theta_\epsilon) - \delta}
\end{align}
\end{proposition}
\begin{proof}
The system dynamics can be written as
\begin{align}
d \ln \theta^t = \frac{d \theta^t}{\theta^t} = [(1-\theta^t)\beta a^*(\theta^t) - \delta]dt
\end{align}
Since $\ln \theta^t$ is monotonically increasing in $\theta \in (0, 1)$, $\theta^0$ evolving to $\theta_\epsilon$ is equivalent to $\ln \theta^0$ evolving to $\ln \theta_\epsilon$. Since $(1-\theta^t)\beta a^*(\theta^t) -\delta$ is decreasing in $\theta^t$, the rate of absolute change $|(1-\theta^t)\beta a^*(\theta^t) -\delta|$ is larger if $\theta^t$ is further away from $\theta^\dagger$. Therefore, before $\ln \theta^t$ decreases (or increases) to $\ln \theta_\epsilon$, the rate of change is at least $(1-\theta_\epsilon)\beta a^*(\theta_\epsilon) - \delta$ and at most $(1-\theta^0)\beta a^*(\theta^0) - \delta$. This proves the proposition.
\end{proof}

\subsection{Attack with Strategic UEs and Unobservable System Compromise State}
In practice, neither the opertaor nor UEs observe the system compromise state in real time. In this case, UEs have to conjecture how the other UEs will participate in the D2D offloading system and the resulting system compromise state. To handle this situation, we formulate a population game to understand the D2D participation incentives under the infectious attacks. To enable tractable analysis, we assume that there are $K$ types of UEs. This is not a strong assumption since we do not impose any restriction on the value of $K$. UEs of the same type $k$ have the same evaluation function $v_{(k)}(\cdot)$, service provision cost $c_{(k)}$ and recovery cost $d_{(k)}$. Denote the fraction of type $k$ UEs by $w_k$ and we must have $\sum_{k=1}^K w_k = 1$.

In the D2D offloading participation game, each UE is a player who decides its participation level. Since UEs do not observe the system compromise state, it is natural to assume the each UE adopts a constant strategy, namely $a^t_i = a_i, \forall t$. The Nash equilibrium of the D2D participation game is defined as follows.

\begin{definition} (Nash Equilibrium).
A participation action profile $\a^\text{NE}$ is a Nash equilibrium if for all $i$, $a_i^\text{NE} = \arg\max_{a_i} U(a_i, \theta^\infty(\a^\text{NE}))$ where $\theta^\infty(\a^\text{NE})$ is the converged system compromise state under $\a^\text{NE}$.
\end{definition}

First, we characterize the converged system compromise state assuming that all type $k$ UEs choose a fixed participation level $a_{(k)}$.

\begin{proposition}
Given the type-wise participation level vector $(a_{(1)},...,a_{(K)})$, there exists a critical effective infection rate $\tau_c = \frac{1}{\sum_{k=1}^K w_k a_{(k)}}$, such that if $\tau \leq \tau_c$, $\theta^\infty = 0$; otherwise, $\theta^\infty > 0$ is the unique solution of $\sum_{k=1}^K \frac{ \tau w_k a_{(k)}}{\tau \theta^\infty a_{(k)} + 1} = 1$.
\end{proposition}
\begin{proof}
Let $\theta_{(k)}$ be the fraction of compromised UEs among all type $k$ UEs. In the steady state, we have the following relation
\begin{align}\label{ss1}
\theta^{\infty}_{(k)} = \frac{\delta^{-1}}{\delta^{-1} + (\theta^\infty\beta a_{(k)})^{-1}} = \frac{\tau\theta^{\infty} a_{(k)}}{\tau\theta^{\infty} a_{(k)} + 1}
\end{align}
where the fraction of the compromised UEs among all UEs is $\theta^\infty = \sum_k w_k \theta^\infty_{(k)}$. It is clear from the above equation that if $\theta^\infty = 0$, then $\theta^\infty_{(k)} = 0, \forall k$. Hence, $\theta^\infty = 0$ is a trivial solution in which $a_{(k)}, \forall k$ can be any value. We now study the non-trivial solution $\theta^\infty > 0$. Rearranging the above equation, we have $\theta^\infty_{(k)} = (1-\theta^\infty_{(k)})\theta^\infty \tau a_{(k)}$. Summing up over $k$ and multiplying by $w_k$ on both sides, we have
\begin{align}\label{fixed}
\theta^\infty =\sum_{k=1}^K w_k \theta^\infty_{(k)} = \tau \theta^\infty \sum_{k=1}^K w_k(1-\theta^\infty_{(k)}) a_{(k)}
\end{align}
This leads to
\begin{align}\label{ss2}
\tau \sum_{k=1}^K w_k(1-\theta^\infty_{(k)}) a_{(k)} = 1
\end{align}
If $\tau \leq \tau_c = \frac{1}{\sum_{k=1}^K w_k a_{(k)}}$, then clearly there is no non-trivial solution of $\theta^\infty_{(k)}$ of the above equation. This implies that the only solution is $\theta^\infty = \theta^\infty_{(k)} = 0, \forall k$, which proves the first half of this proposition. Next, we show that if $\tau > \tau_c$, there indeed exists a unique solution $\theta^\infty > 0$. Substituting \eqref{ss1} into \eqref{ss2} yields
\begin{align}
 \sum_{k=1}^K \frac{\tau w_k a_{(k)}}{\tau \theta^\infty a_{(k)} + 1} = 1
\end{align}
Clearly the left-hand side of the above equation $\text{LHS}(\theta^\infty)$ is decreasing in $\theta^\infty$. Moreover $\text{LHS}(0) = \tau \sum_{k=1}^K w_k a_{(k)} > 1$, and $\text{LHS}(1) = \tau \sum_{k=1}^K \frac{w_k a_{(k)}}{\tau a_{(k)} + 1} < \sum_{k=1}^K w_k = 1$. Therefore, there is a unique solution of $\theta^\infty \in (0, 1)$.
\end{proof}
Proposition 6 is actually the extended version of Proposition 3, which further considers heterogeneous UEs. In the homogeneous UE case, the critical infection rate depends on the homogeneous participation level of UEs. In the heterogeneous UE case, the critical infection rate depends on the average participation level of UEs. With Proposition 6 in hand, we are able to characterize the Nash equilibrium of the D2D participation game.

\begin{theorem}
(1) The D2D participation game has a unique NE. (2) The NE is symmetric within each type, namely $a^\text{NE}_i = a^\text{NE}_{(k)}$ for all UE $i$ with type $k$. (3) If $\tau \leq \frac{1}{\sum_{k=1}^K w_k a^\text{AF}_{(k)}}$, then $\theta^\infty = 0$ and $a^\text{NE}_{(k)} = a^\text{AF}_{(k)}, \forall k$. Otherwise, $\theta^\infty > 0$ and $\sum_{k=1}^K w_k a^\text{NE}_{(k)} > \tau^{-1}$.
\end{theorem}
\begin{proof}
Consider type $k$ UEs. Each UE chooses the individual optimal participation level determined by the following equation
\begin{align}
u'_{(k)}(a_i)(\rho + \delta \beta\theta^\infty a_i) - u_{(k)}(a_i)\beta \theta^\infty - \beta\theta^\infty q_{(k)} = 0
\end{align}
Given the same $\beta,\delta, \theta^\infty$, there is a unique optimal solution $a^*_i$ according to Proposition 1. Therefore, if an equilibrium exists, UEs of the same type must choose the same participation level. To prove the existence of NE is to prove that the following function has a fixed point $\theta^\infty$ based on our analysis in the proof of Proposition 6:
\begin{align}\label{fixed2}
\theta^\infty = \tau \theta^\infty \sum_{k=1}^K w_k(1-\theta^\infty_{(k)}) a_{(k)}(\theta^\infty)
\end{align}
Note that the difference from \eqref{fixed} is that $a_{(k)}(\theta^\infty)$ is a function of $\theta^\infty$ rather than a prescribed action.

First, we investigate if $\theta^\infty = 0$ could be a fixed point. If $\theta^\infty = 0$, then $a^*_{(k)} = a^\text{AF}_{(k)}, \forall k$, which is the optimal participation level in the attack-free network. Therefore, if $\tau \leq \frac{1}{\sum_{k=1}^K w_k a^\text{AF}_{(k)}}$, then $\theta^\infty = 0$ is a fixed point. Otherwise, $\theta^\infty = 0$ is not a fixed point.

Next, we investigate if there is any $\theta^\infty > 0$ that can be fixed point. This is to show, according to \eqref{fixed2}, if there is a solution to
\begin{align}
\sum_{k=1}^K \frac{ w_k \tau  a_{(k)}(\theta^\infty) }{\tau \theta^\infty a_{(k)}(\theta^\infty) + 1} = 1
\end{align}
Denote the left-hand side function by $f(\theta^\infty)$.
\begin{align}
f'(\theta^\infty) = \sum_{k=1}^K w_k \tau\frac{a'_{(k)}(\theta^\infty) - \tau a^2_{(k)}(\theta^\infty) }{(\tau \theta^\infty a_{(k)}^*(\theta^\infty) + 1)^2} < 0
\end{align}
The inequality is because $a_{(k)}(\theta^\infty)$ decreases with $\theta^\infty$ according to Proposition 1. Moreover, $f(1) < \sum_{k=1}^K w_k = 1$ and $f(0) = \tau \sum_{k=1}^K w_k  a^\text{AF}_{(k)}$. Therefore, if $\tau > \frac{1}{\sum_{k=1}^K w_k  a^\text{AF}_{(k)}}$, then $f(0) > 1$. This means that there exists a unique positive solution $\theta^\infty$.
\end{proof}

\begin{corollary}
If UEs are homogeneous, i.e. $K = 1$, then the D2D computing participation game has a unique symmetric NE. Moreover, if $\tau \leq \frac{1}{a^\text{AF}}$, then $\theta^\infty = 0$ and $a^{NE} = a^{AF}$; otherwise, $\theta^\infty = \theta^\dagger$ where $\theta^\dagger$ has the same value given in Proposition 4.
\end{corollary}

Theorem 1 and Corollary 1 show that, even if UEs do not observe $\theta^t$ in real-time, the system will converge to the same state when $\theta^t$ is observed as established in Proposition 4. The latter case indeed can be considered as that UEs are playing best response dynamics of the population game, which converges to the predicted NE. Moreover, the thresholding effect still exhibits. If the effective infection rate is sufficiently small, then the infectious attacks extinguish. Otherwise, the infectious attacks persist.

\section{Optimal Reward Mechanism Design}
We are now ready to design the optimal reward mechanism. Under the assumption that there are $K$ types of UEs, the operator's problem can be written as
\begin{align}\label{opt_virus}
\max_{r_0} ~(b_0 - r_0)\sum_{k=1}^K w_k (1-\theta_{(k)}^\infty) a^*_{(k)}(\theta^\infty)
\end{align}
Note that $\theta^\infty$, $\theta_{(k)}^\infty$ and $a^*_{(k)}$ all depend on the reward mechanism $r_0$ even though the dependency is not explicitly expressed. Solving the above optimization problem is difficult because there are no closed-form solutions of $\theta^\infty$, $\theta_{(k)}^\infty$ and $a^*_{(k)}$ in terms of $r_0$ since they are complexly coupled as shown in the previous sections. Fortunately, our analysis shows that there is a structural property that we can exploit to solve this problem in a much easier way.


\begin{theorem}
The optimal reward mechanism design problem $\eqref{opt_virus}$ under infectious attacks is equivalent to the following constrained reward design problem in the attack-free network
\begin{align}\label{opt_r}
\max_{r_0} & ~~~(b_0 - r_0)\sum_{k=1}^K w_k a^\text{AF}_{(k)}(r_0)\\
\textit{s.t.} & ~~~ \sum_{k=1}^K w_k a^\text{AF}_{(k)}(r_0) \leq \tau^{-1}
\end{align}
\end{theorem}
\begin{proof}
We divide the reward mechanism $r_0$ into two categories $\mathcal{R}_1$ and $\mathcal{R}_2$. Consider any reward mechanism $r_0$, if the resulting $\sum_{k=1}^K a^*_{(k)}(\theta^\infty) \geq \tau^{-1}$, then $r_0 \in \mathcal{R}_1$. Otherwise $r_0 \in \mathcal{R}_2$.

Now, according to Theorem 1, if $r_0 \in \mathcal{R}_1$, then we also have $\sum_{k=1}^K w_{(k)}(1-\theta^\infty) a^*_{(k)}(\theta^\infty) = \tau^{-1}$, which is a constant that does not depend on the exact value of $r_0$. Therefore, the optimal $r_0$ in $\mathcal{R}_1$ must be the smallest possible $r_0$ in order to maximize the operator's utility. The smallest $r_0$ is the one such that $\sum_{k=1}^K w_k a^*_{(k)}(\theta^\infty) = \tau^{-1}$ and $\theta^\infty = 0$. Since $\theta^\infty = 0$, $\sum_{k=1}^K w_k a^*_{(k)}(\theta^\infty) = \tau^{-1}$ is equivalent to $\sum_{k=1}^K w_k a^\text{AF}_{(k)} = \tau^{-1}$. This means that if $r_0 \in \mathcal{R}_1$ is the optimal solution, it is also a feasible solution of the above constrained optimization problem.

If $r_0 \in \mathcal{R}_2$, then according to Theorem 1, we have $\theta^\infty = 0$. Again, since $\theta^\infty = 0$,  $\sum_{k=1}^K a^*_{(k)}(\theta^\infty) < \tau^{-1}$ is equivalent to $\sum_{k=1}^K w_k a^\text{AF}_{(k)} < \tau^{-1}$. This also proves that if $r_0 \in \mathcal{R}_2$ is the optimal solution, it is also a feasible solution of the above constrained optimization problem.
\end{proof}
Theorem 2 converts the reward optimization problem in the presence of infectious attack risks into an optimization problem in the attack-free network by simply adding a constraint. Because $a^\text{AF}_{(k)}(r_0)$ can be easily computed, the converted optimization problem can be easily solved through numerical methods. In fact, since $a^\text{AF}_{(k)}(r_0)$ is increasing in $r_0$, the search space of the optimal $r_0$ can be restricted to $[0, \bar{r}_0]$ where $\bar{r}_0$ makes the constraint binding, i.e. $\sum_{k=1}^K w_k a^\text{AF}_{(k)}(\bar{r}_0) = \tau^{-1}$.

\begin{figure}
  \centering
  \includegraphics[width=2.8in]{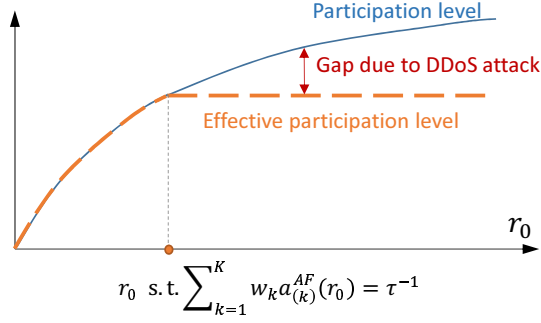}\\
  \caption{Gap between participation level and effective participation level caused by attacks.}\label{gap}
  \vspace{-10pt}
\end{figure}

The intuition of Theorem 2 is that the optimal reward mechanism must not promote too much participation that induce persistent attacks in the network. This is because too much participation does not improve the \textit{effective} participation level due to UEs compromised by the attack (see an illustration in Figure \ref{gap}). Since less participation requires a smaller unit reward $r_0$, more utility can be obtained for the operator by employing a smaller unit reward. We call this the ``less is more'' phenomenon in the D2D offloading network under infectious attack risks. Corollary 2 further compares the optimal reward mechanisms with and without the infectious attack risks.

\begin{corollary} The optimal reward mechanism $r^*$ for D2D offloading under infectious attack risks is no more than the optimal reward mechanism $r^\text{AF*}$ in the attack-free network.
\end{corollary}
\begin{proof}
This is a direct result of Theorem 2. If $r^\text{AF*} \in \mathcal{R}_2$, then $r^* = r^\text{AF*}$. If $r^\text{AF*} \in \mathcal{R}_1$, then $r^* < r^\text{AF*}$.
\end{proof}

Theorem 2 also implies that a larger effective infection rate reduces the operator's utility.

\begin{corollary}
The optimal utility of the operator is non-increasing in $\tau$.
\end{corollary}
\begin{proof}
This is a direct result of Theorem 2 since a larger $\tau$ imposes a more stringent constraint in the optimization problem \eqref{opt_r}.
\end{proof}

Corollary 3 implies that the operator's utility can be improved by reducing the effective infection rate $\tau$. This can be done by developing and deploying better security technologies that either reduce the attack success rate $\beta$ or improve the recovery rate $\delta$. Hence, the operator may consider jointly optimize the reward mechanism $r_0$ and security technology that results in a smaller $\tau$. This is to solve the following optimization problem,
\begin{align}\label{opt_r_tau}
\max_{r_0, \tau} & ~~~(b_0 - r_0)\sum_{k=1}^K w_k a^\text{AF}_{(k)}(r_0) - J(\tau)\\
\textit{s.t.} & ~~~ \sum_{k=1}^K w_k a^\text{AF}_{(k)}(r_0) \leq \tau^{-1}
\end{align}
where $J(\tau)$ is the technology development cost that achieves an effective infective rate $\tau$. Typically, $J(\tau)$ is decreasing in $\tau$. It is evident that the optimal solution must have $\sum_{k=1}^K w_k a^\text{AF}_{(k)}(r^*_0) = (\tau^*)^{-1}$. This is because if $\sum_{k=1}^K w_k a^\text{AF}_{(k)}(r^*_0) < (\tau^*)^{-1}$, then we can always choose a larger $\tilde{\tau} > \tau^*$ such that $J(\tilde{\tau}) < J(\tau^*)$. This leads to a higher utility which contradicts the optimality of $\tau^*$. Therefore, the joint optimization problem reduces to an unconstrained problem as follows
\begin{align}
\max_{r_0} & ~~~(b_0 - r_0)\sum_{k=1}^K w_k a^\text{AF}_{(k)}(r_0) - J(\frac{1}{\sum_{k=1}^K w_k a^\text{AF}_{(k)}(r_0)})
\end{align}
This problem can be easily solved using numerical methods.

\section{Simulations}
\subsection{Simulation Setup and Time System Conversion}
Since our analytical model adopts a continuous time system, we divide time into small time slots to enable the simulation. Specifically, a unit time in the continuous time system is divided into 100 slots and we simulate a large number of slots. We simulate a number $N = 100$ of mobile UEs moving in a square area of size $100\times 100$.  User mobility follows the random waypoint model. Specifically, when the UE is moving, it moves at a random speed between 0 and $v_\text{max}$ per slot towards a randomly selected destination. When the UE reaches the destination, it pauses for a random number of slots between 0 and $m_\text{max}$ and then selects a new destination. The parameters $v_\text{max}$ and $m_\text{max}$ control the mobility level of the network and will be varied to study the sensitivity of the random server-requester matching approximation to the actual UE mobility and server-requester matching.

\begin{figure}
  \centering
  \includegraphics[width=3in]{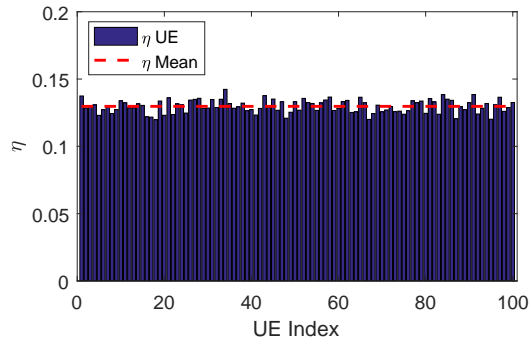}\\
  \vspace{-10pt}
  \caption{Estimation of offloading request arrival.}\label{eta}
\end{figure}

In every slot, a UE has computation tasks to offload and hence becomes a requester with probability $p$. When a UE is not a requester, it is able to provide D2D offloading service. Therefore, with probability $1-p$ the UE is a potential server. The number of tasks of each requester in every slot is randomly selected between 1 and $W_{max}$. For each requester, we find the potential server UEs within a distance $d$ of the requester where $d$ is the D2D communication range. Suppose that UEs are obedient, then the operator assigns one task to one of these UEs. The probability $\eta$ that a potential server UE receives an offloading request in the obedient UE case can be estimated in simulations. Figure \ref{eta} illustrates the estimations for parameters $p = 0.2$, $v_\text{max} = 20$. Estimating this probability is very important for the conversion of the participation level in continuous time into its counterpart in discrete time. Specifically, a chosen participation level $a$ per unit time in continuous time is converted to a participation probability $\frac{a}{100(1-p)\eta}$ in each slot when the UE is a potential server in discrete time. With this conversion, in the strategic UE case, the operator assigns one task to one of the potential server UE with probability $\frac{a}{100(1-p)\eta}$. The the remaining tasks, if any, are offloaded to the edge server on the BS. UEs can modify their participation decisions in every slot. Moreover, the D2D offloading benefit $b$ and cost $c_i$ in continuous time also have to be converted to their discrete time counterpart. The conversion is also performed for the recovery process and recovery cost.

\subsection{System Dynamics}

\begin{figure}
\centering
\begin{minipage}[t]{.5\textwidth}
  \centering
  \includegraphics[width=0.9\linewidth]{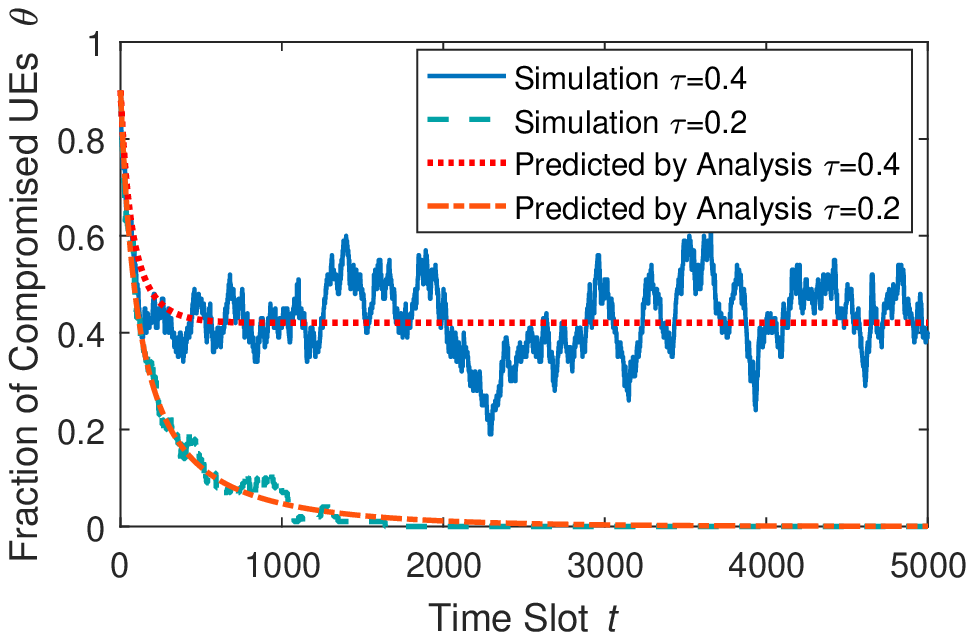}
  \caption{Epidemic dynamics for non-strategic UE.}\label{predetermined_two_type_theta}
\end{minipage}%
\begin{minipage}[t]{.5\textwidth}
  \centering
  \includegraphics[width=0.9\linewidth]{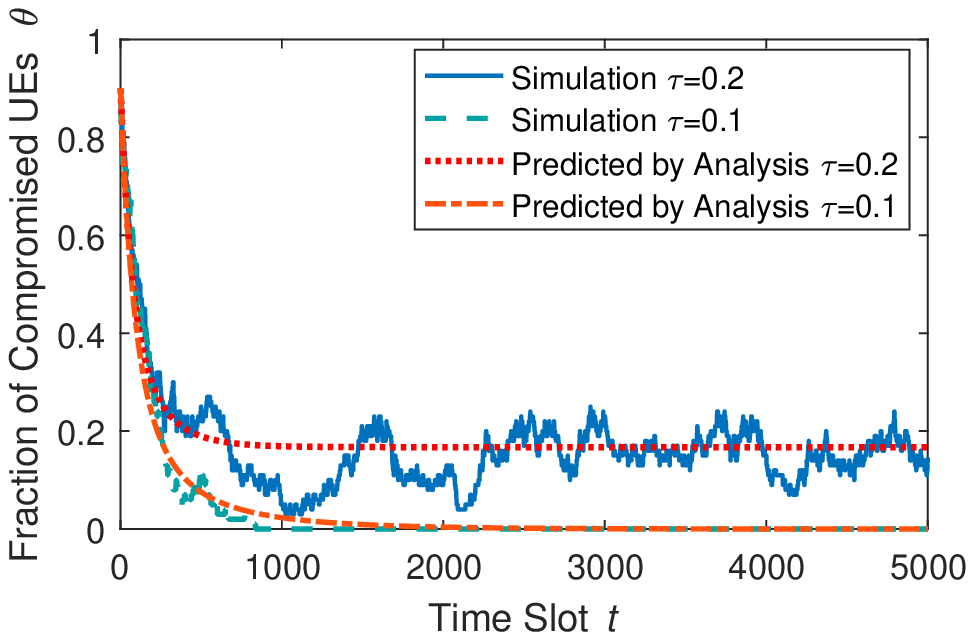}
  \caption{Epidemic dynamics for strategic UEs.}\label{strategic_two_type_theta}
\end{minipage}%
\end{figure}

Figure \ref{predetermined_two_type_theta} illustrates the system dynamics for non-strategic UEs who follow a prescribed participation strategy. Two types of UEs are considered in this simulation. Type 1 UEs adopt a participation level $a_{(1)} = 3$ and Type 2 UEs adopt a participation level $a_{(2)} = 5$. The fractions of these two types of users are $w_{(1)} = 0.3$ and $w_{(2)} = 0.7$, respectively. Therefore, the predicted critical effective infection rate $\tau_c = 0.227$ according to Proposition 6. We fix $\delta = 1$ and show the results for $\beta = 0.2$ and $\beta = 0.4$, which correspond to $\tau = 0.2$ and $\tau = 0.4$, respectively. As shown in Figure \ref{predetermined_two_type_theta}, when $\tau < \tau_c$, the infections extinguish over time. When $\tau > \tau_c$, infections persist in the system at a compromise level around 0.42. Because we used a relatively small finite UE population in the simulation, there are still fluctuations in the results. However, the predicted dynamics by our analysis is in accordance to the simulation results and the fluctuations will be less with a larger UE population.

Figure \ref{strategic_two_type_theta} illustrates the system dynamics for strategic UEs for the same setting as above. The difference is that, since UEs are strategic, they will decide their participation levels by themselves. The user evaluation functions are chosen as $v_{(1)}(x) = \sqrt{x}$ and $v_{(2)} = 1.5\sqrt{x}$. The costs for users are the same $c = 0.35$. The reward offered by the operator is set as $r = 2.2$. Therefore, the critical infection rate is computed to be $\tau_c = 0.12$ according to Theorem 1. As shown in Figure \ref{strategic_two_type_theta}, infections extinguish over time when $\tau = 0.1$ which is smaller than $\tau_c$ and persist when $\tau = 0.2$ which is larger than $\tau_c$. Notice that in the non-strategic UE case, $\tau = 0.2$ will instead make attacks extinguish.

\begin{figure}
  \centering
  \includegraphics[width=3in]{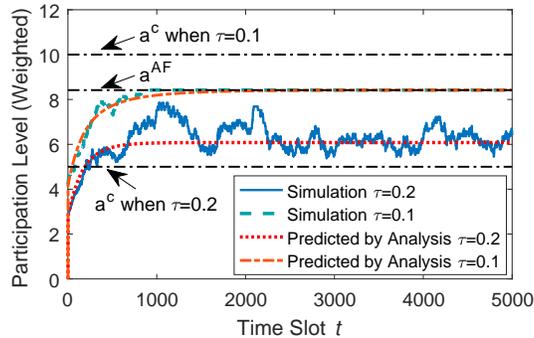}\\
  \vspace{-10pt}
  \caption{Evolution of UE participation levels.}\label{strategic_two_type_weight}
\end{figure}

To better understand what is happening behind this epidemic dynamics, we show the evolution of UE participation levels (weighted average of the two types) in Figure \ref{strategic_two_type_weight}. When the effective infection rate is lower, UEs tend to choose a higher participation level. Regardless of the exact value of $\tau$, the converged participation level does not exceed the optimal participation level $a^{\text{AF}}$ in the attack-free network (which does not depend on $\tau$). When $\tau = 0.2$, the converged value is greater than the corresponding critical participation level $a^c$ (which depends on $\tau$) and hence, the attacks persist. When $\tau = 0.1$, the converged value is smaller than the corresponding critical participation level $a^c$, thereby eliminating the attacks.

\subsection{Optimal Reward Mechanism}

\begin{figure}
\centering
\begin{minipage}[t]{.48\textwidth}
  \centering
  \includegraphics[width=0.9\linewidth]{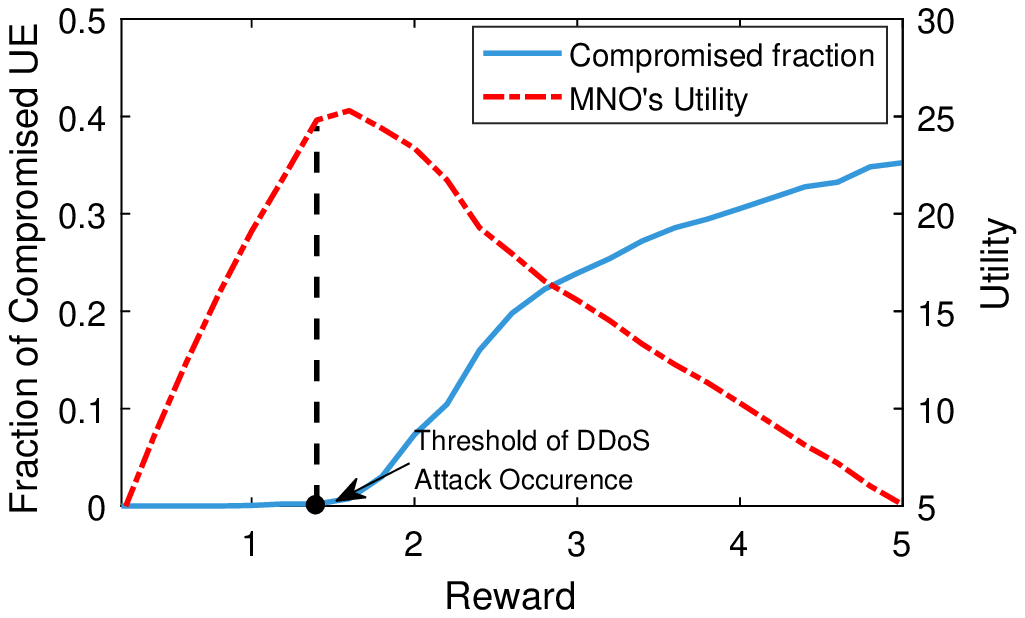}
  \caption{Impact of reward on system compromise state and the operator's utility.}\label{reward_theta_Utility}
\end{minipage}
\begin{minipage}[t]{.48\textwidth}
  \centering
  \includegraphics[width=0.9\linewidth]{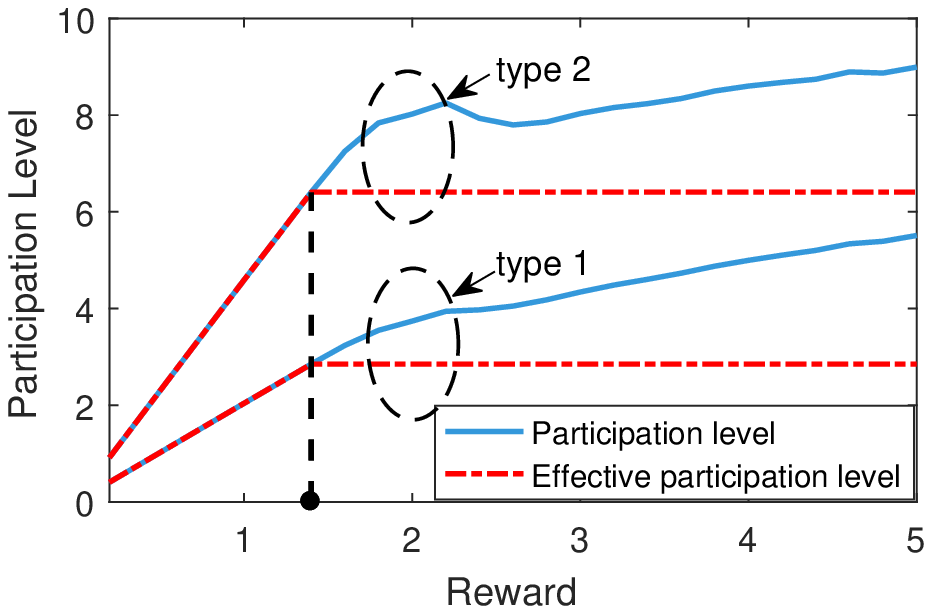}
  \caption{Impact of reward on effective participation level.}\label{reward_participation_level}
\end{minipage}%
\end{figure}

Now we consider the operator's utility maximization problem. In this set of simulations, the benefit for the operator is set as $b = 6$. Figure \ref{reward_theta_Utility} shows the impact of the reward on the fraction of compromised UEs in the network as well as the operator's utility. As the reward increases, UEs have more incentives to participate and when the reward increases to a certain point, infections become persistent. As a result, further increasing the reward $r_0$ decreases the operator's utility since the effective participation level of UEs does not improve as shown in Figure \ref{reward_participation_level}. This is the predicted ``Less is More'' phenomenon. This result is significantly important for the operator to determine the optimal reward mechanism that is security-aware. Figures \ref{tau_reward_utility} and \ref{tau_reward_theta} further show the operator's utility and the fraction of compromised UEs as functions of the reward $r_0$ and the effective infection rate $\tau$. Again, the simulation results are in accordance with our analytical results.

\begin{figure}
\centering
\begin{minipage}[t]{.48\textwidth}
  \centering
  \includegraphics[width=0.9\linewidth]{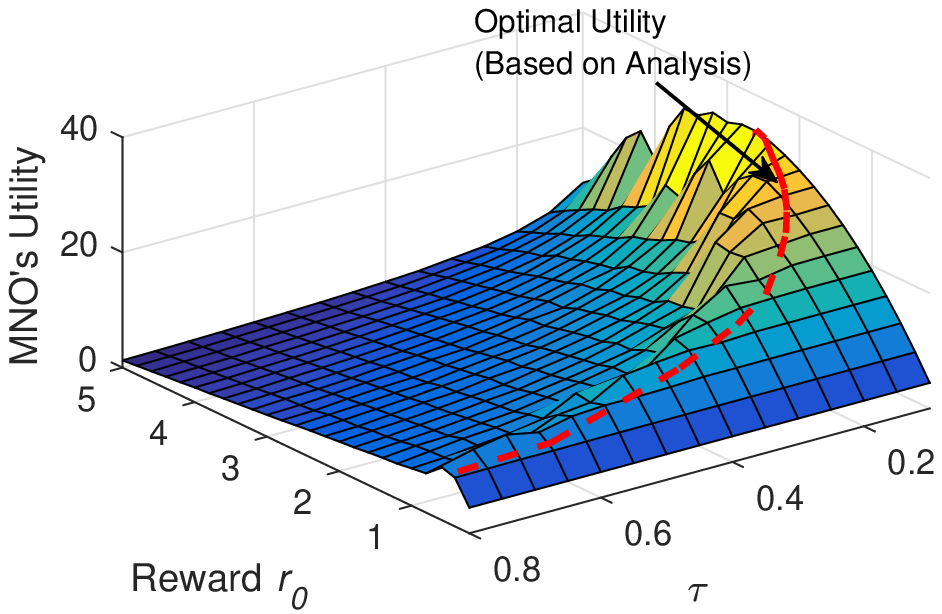}
  \caption{Operator's utility as a function of $r_0$ and $\tau$.}\label{tau_reward_utility}
\end{minipage}
\begin{minipage}[t]{.48\textwidth}
  \centering
  \includegraphics[width=0.9\linewidth]{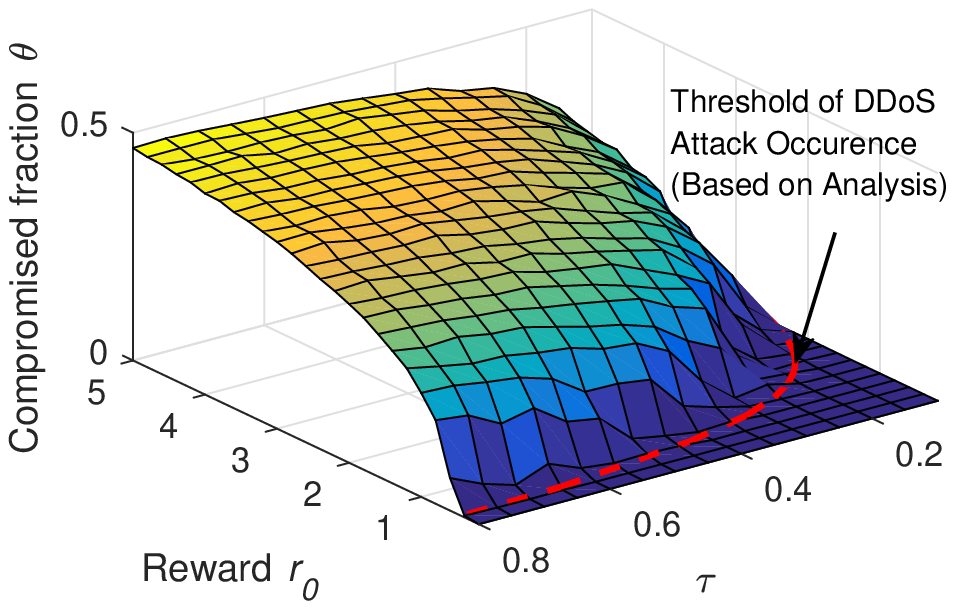}
  \caption{System compromise state as a function of $r_0$ and $\tau$.}\label{tau_reward_theta}
\end{minipage}%
\end{figure}

\subsection{Impact of UE Mobility}

\begin{figure}
\centering
\begin{minipage}[t]{.48\textwidth}
  \centering
  \includegraphics[width=0.9\linewidth]{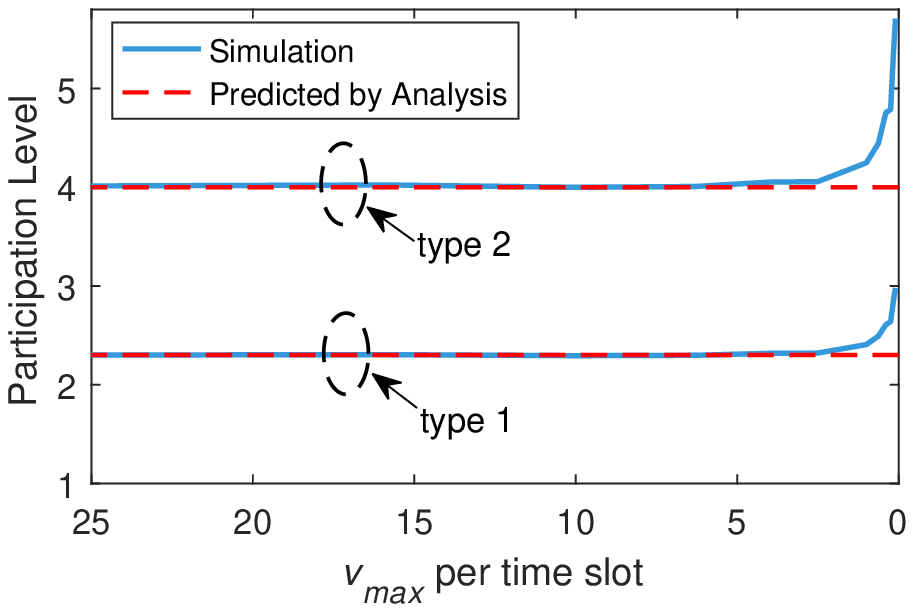}
  \caption{Impact of mobility on the participation levels.}\label{v_participation}
\end{minipage}
\begin{minipage}[t]{.48\textwidth}
  \centering
  \includegraphics[width=0.9\linewidth]{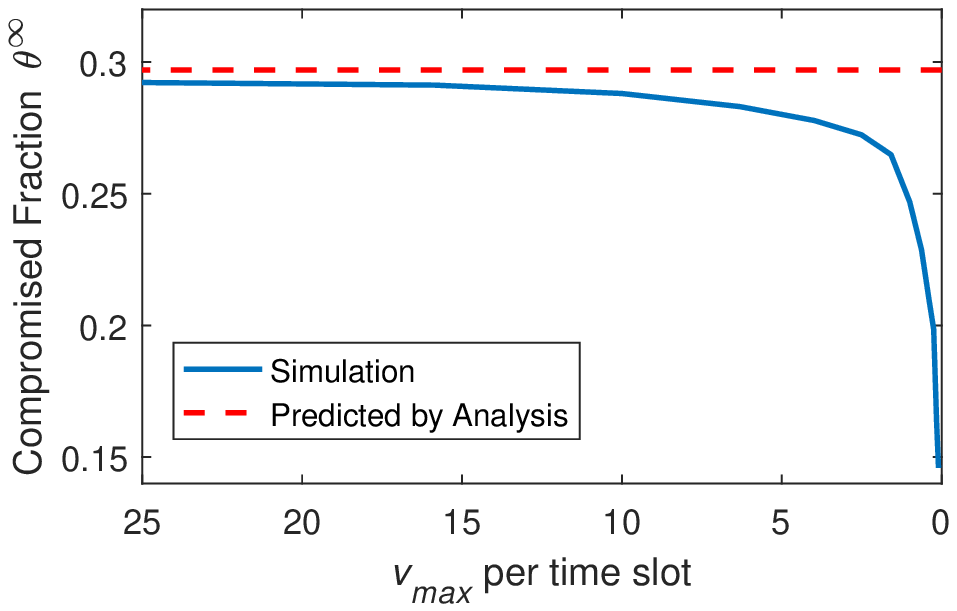}
  \caption{Impact of mobility on the system compromise state.}\label{v_theta}
\end{minipage}%
\end{figure}

Finally, we investigate the impact of UE mobility on the accuracy of our model and analysis by varying the moving speed of UEs. Figures \ref{v_participation} and \ref{v_theta} show how the UE participation levels and the fraction of compromised UEs change with UE mobility, respectively. In our model, we assumed that a UE receives requests from other UEs uniformly randomly, which is a good approximation when UEs' mobility is fast. However, when UEs' mobility is slow, they will more likely have localized interactions with only a subset of UEs with high probability. For instance, in practice, people are more likely to appear in the same locations at the same time with their family, friends and colleagues. As shown in Figures \ref{v_participation} and \ref{v_theta}, when UE mobility is fast, our analytical results are very well aligned with the simulation results. However, when UE mobility is slow, there is an obvious deviation of the simulation results from our analysis, suggesting that new models are needed to handle low mobility network scenarios. This is an interesting future work direction.

\section{Conclusions}
In this paper, we investigated the extremely important but much less studied incentive mechanism design problem in dynamic networks where users' incentives and security risks they face are intricately coupled. We adopted a dynamic non-cooperative game theoretic approach to understand how user collaboration incentives are influenced by interdependent security risks such as the infectious attack risks, and how the attack risks evolve, propagate, persist and extinguish depending on users' strategic choices. This understanding allowed us to develop security-aware incentive mechanisms that are able to combat and mitigate attacks in D2D offloading systems. Our study leverages the classic epidemic models, but on the other hand, it represents a significant departure from these models since users are strategically choosing their actions rather than obediently following certain prescribed rules. Our model and analysis not only provide new and important insights and guidelines for designing more efficient and more secure D2D offloading networks but also can be adapted to solve many other challenging problems in other cooperative networks where users face interdependent security risks. Future work includes investigating user interaction models that are more localized and social network-based, and user heterogeneity in terms of adopted security technologies.


%

%
%
%
%
%

\ifCLASSOPTIONcaptionsoff
  \newpage
\fi


\bibliographystyle{IEEEtran}
\bibliography{refs,supp}

\end{document}

%% file: macros.tex
\newcommand{\bp}{\begin{proof} \small }
\newcommand{\ep}{\end{proof} \normalsize}
\newcommand{\epx}{\end{proof} \small}
\newcommand{\bpa}{\begin{proofappx} \footnotesize }
\newcommand{\epa}{\end{proofappx} \small }
\newtheorem{theorem}{Theorem}
\newtheorem{proposition}{Proposition}
\newtheorem{corollary}{Corollary}

\newtheorem{assumption}{Assumption}
\newtheorem{definition}{Definition}

\newtheorem*{theorem*}{Theorem}
\newtheorem*{proposition*}{Proposition}
\newtheorem*{corollary*}{Corollary}
\newtheorem*{lemma*}{Lemma}
\newtheorem*{assumption*}{Assumption}
\newtheorem*{definition*}{Definition}
\newtheorem*{claim*}{Claim}

\newcommand{\be}{\begin{equation}}
\newcommand{\ee}{\end{equation}}
\newcommand{\bs}{\begin{subequations}}
\newcommand{\es}{\end{subequations}}
\newcommand{\bq}{\begin{eqnarray}}
\newcommand{\eq}{\end{eqnarray}}
\newcommand{\bqn}{\begin{eqnarray*}}
\newcommand{\eqn}{\end{eqnarray*}}

\newcommand{\ba}{\left[ \begin{array}}
\newcommand{\ea}{\\ \end{array} \right]}
\newcommand{\ben}{\begin{enumerate}}
\newcommand{\een}{\end{enumerate}}

\def\a{{\boldsymbol{a}}}

\def\real{{\mathchoice%
{\hbox{\rm\setbox1=\hbox{I}\copy1\kern-.45\wd1 R}}
{\hbox{\rm\setbox1=\hbox{I}\copy1\kern-.45\wd1 R}}
{\hbox{\scriptsize\rm\setbox1=\hbox{I}\copy1\kern-.45\wd1 R}}
{\hbox{\scriptsize\rm\setbox1=\hbox{I}\copy1\kern-.45\wd1 R}}}}

\def\Zint{{\mathchoice{\setbox1=\hbox{\sf Z}\copy1\kern-.75\wd1\box1}
{\setbox1=\hbox{\sf Z}\copy1\kern-.75\wd1\box1}
{\setbox1=\hbox{\scriptsize\sf Z}\copy1\kern-.75\wd1\box1}
{\setbox1=\hbox{\scriptsize\sf Z}\copy1\kern-.75\wd1\box1}}}
\newcommand{\complex}{ \hbox{\rm C\kern-0.45em\rule[.07em]{.02em}{.58em}%
\kern 0.43em}}

%% file: D2DMEC-arxiv.bbl
\begin{thebibliography}{10}
\providecommand{\url}[1]{#1}
\csname url@samestyle\endcsname
\providecommand{\newblock}{\relax}
\providecommand{\bibinfo}[2]{#2}
\providecommand{\BIBentrySTDinterwordspacing}{\spaceskip=0pt\relax}
\providecommand{\BIBentryALTinterwordstretchfactor}{4}
\providecommand{\BIBentryALTinterwordspacing}{\spaceskip=\fontdimen2\font plus
\BIBentryALTinterwordstretchfactor\fontdimen3\font minus
  \fontdimen4\font\relax}
\providecommand{\BIBforeignlanguage}[2]{{%
\expandafter\ifx\csname l@#1\endcsname\relax
\typeout{** WARNING: IEEEtran.bst: No hyphenation pattern has been}%
\typeout{** loaded for the language `#1'. Using the pattern for}%
\typeout{** the default language instead.}%
\else
\language=\csname l@#1\endcsname
\fi
#2}}
\providecommand{\BIBdecl}{\relax}
\BIBdecl

\bibitem{dinh2013survey}
H.~T. Dinh, C.~Lee, D.~Niyato, and P.~Wang, ``A survey of mobile cloud
  computing: architecture, applications, and approaches,'' \emph{Wireless
  communications and mobile computing}, vol.~13, no.~18, pp. 1587--1611, 2013.

\bibitem{mao2017mobile}
Y.~Mao, C.~You, J.~Zhang, K.~Huang, and K.~B. Letaief, ``A survey on mobile
  edge computing: The communication perspective,'' \emph{IEEE Communications
  Surveys \& Tutorials}, 2017.

\bibitem{chiang2016fog}
M.~Chiang and T.~Zhang, ``Fog and iot: An overview of research opportunities,''
  \emph{IEEE Internet of Things Journal}, vol.~3, no.~6, pp. 854--864, 2016.

\bibitem{marinelli2009hyrax}
E.~E. Marinelli, ``Hyrax: cloud computing on mobile devices using mapreduce,''
  DTIC Document, Tech. Rep., 2009.

\bibitem{shi2012serendipity}
C.~Shi, V.~Lakafosis, M.~H. Ammar, and E.~W. Zegura, ``Serendipity: enabling
  remote computing among intermittently connected mobile devices,'' in
  \emph{Proceedings of the thirteenth ACM international symposium on Mobile Ad
  Hoc Networking and Computing}.\hskip 1em plus 0.5em minus 0.4em\relax ACM,
  2012, pp. 145--154.

\bibitem{huerta2010virtual}
G.~Huerta-Canepa and D.~Lee, ``A virtual cloud computing provider for mobile
  devices,'' in \emph{Proceedings of the 1st ACM Workshop on Mobile Cloud
  Computing \& Services: Social Networks and Beyond}.\hskip 1em plus 0.5em
  minus 0.4em\relax ACM, 2010, p.~6.

\bibitem{li2014can}
Y.~Li and W.~Wang, ``Can mobile cloudlets support mobile applications?'' in
  \emph{Infocom, 2014 proceedings ieee}.\hskip 1em plus 0.5em minus 0.4em\relax
  IEEE, 2014, pp. 1060--1068.

\bibitem{mtibaa2015friend}
A.~Mtibaa, K.~Harras, and H.~Alnuweiri, ``Friend or foe? detecting and
  isolating malicious nodes in mobile edge computing platforms,'' in
  \emph{Cloud Computing Technology and Science (CloudCom), 2015 IEEE 7th
  International Conference on}.\hskip 1em plus 0.5em minus 0.4em\relax IEEE,
  2015, pp. 42--49.

\bibitem{lu2014can}
Z.~Lu, W.~Wang, and C.~Wang, ``How can botnets cause storms? understanding the
  evolution and impact of mobile botnets,'' in \emph{INFOCOM, 2014 Proceedings
  IEEE}.\hskip 1em plus 0.5em minus 0.4em\relax IEEE, 2014, pp. 1501--1509.

\bibitem{lu2016evolution}
------, ``On the evolution and impact of mobile botnets in wireless networks,''
  \emph{IEEE Transactions on Mobile Computing}, vol.~15, no.~9, pp. 2304--2316,
  2016.

\bibitem{fudenberg1991game}
D.~Fudenberg and J.~Tirole, ``Game theory, 1991,'' \emph{Cambridge,
  Massachusetts}, vol. 393, 1991.

\bibitem{kermack1927contribution}
W.~O. Kermack and A.~G. McKendrick, ``A contribution to the mathematical theory
  of epidemics,'' in \emph{Proceedings of the Royal Society of London A:
  mathematical, physical and engineering sciences}, vol. 115, no. 772.\hskip
  1em plus 0.5em minus 0.4em\relax The Royal Society, 1927, pp. 700--721.

\bibitem{fernando2013mobile}
N.~Fernando, S.~W. Loke, and W.~Rahayu, ``Mobile cloud computing: A survey,''
  \emph{Future Generation Computer Systems}, vol.~29, no.~1, pp. 84--106, 2013.

\bibitem{chatzopoulos2016have}
D.~Chatzopoulos, M.~Ahmadi, S.~Kosta, and P.~Hui, ``Have you asked your
  neighbors? a hidden market approach for device-to-device offloading,'' in
  \emph{World of Wireless, Mobile and Multimedia Networks (WoWMoM), 2016 IEEE
  17th International Symposium on A}.\hskip 1em plus 0.5em minus 0.4em\relax
  IEEE, 2016, pp. 1--9.

\bibitem{huang2012dynamic}
D.~Huang, P.~Wang, and D.~Niyato, ``A dynamic offloading algorithm for mobile
  computing,'' \emph{IEEE Transactions on Wireless Communications}, vol.~11,
  no.~6, pp. 1991--1995, 2012.

\bibitem{satyanarayanan2009case}
M.~Satyanarayanan, P.~Bahl, R.~Caceres, and N.~Davies, ``The case for vm-based
  cloudlets in mobile computing,'' \emph{IEEE pervasive Computing}, vol.~8,
  no.~4, pp. 14--23, 2009.

\bibitem{chang2017collaborative}
Z.~Chang, S.~Zhou, T.~Ristaniemi, and Z.~Niu, ``Collaborative mobile clouds: An
  energy efficient paradigm for content sharing,'' \emph{IEEE Wireless
  Communications}, 2017.

\bibitem{chen2015computation}
M.~Chen, Y.~Hao, Y.~Li, C.-F. Lai, and D.~Wu, ``On the computation offloading
  at ad hoc cloudlet: architecture and service modes,'' \emph{IEEE
  Communications Magazine}, vol.~53, no.~6, pp. 18--24, 2015.

\bibitem{Min2011}
H.~Min, J.~Lee, S.~Park, and D.~Hong, ``Capacity enhancement using an
  interference limited area for device-to-device uplink underlaying cellular
  networks,'' \emph{IEEE Transactions on Wireless Communications}, vol.~10,
  no.~12, pp. 3995--4000, Dec. 2011.

\bibitem{Feng2015}
D.~Feng, G.~Yu, C.~Xiong, Y.~Yuan-Wu, G.~Y. Li, G.~Feng, and S.~Li, ``Mode
  switching for energy-efficient device-to-device communications in cellular
  networks,'' \emph{IEEE Transactions on Wireless Communications}, vol.~14,
  no.~12, pp. 6993--7003, Dec. 2015.

\bibitem{Poor2014}
L.~Al-Kanj, H.~V. Poor, and Z.~Dawy, ``Optimal cellular offloading via
  device-to-device communication networks with fairness constraints,''
  \emph{IEEE Transactions on Wireless Communications}, vol.~13, no.~8, pp.
  4628--4643, Aug. 2014.

\bibitem{Asadi2014}
A.~Asadi, Q.~Wang, and V.~Mancuso, ``A survey on {Device-to-Device}
  communication in cellular networks,'' \emph{IEEE Communications Surveys
  Tutorials}, vol.~16, no.~4, pp. 1801--1819, Fourthquarter 2014.

\bibitem{FlashLinQ}
X.~Wu, S.~Tavildar, S.~Shakkottai, T.~Richardson, J.~Li, R.~Laroia, and
  A.~Jovicic, ``{FlashLinQ}: A synchronous distributed scheduler for
  peer-to-peer ad hoc networks,'' \emph{IEEE/ACM Transactions on Networking},
  vol.~21, no.~4, pp. 1215--1228, Aug. 2013.

\bibitem{Lin2014}
X.~Lin, J.~G. Andrews, A.~Ghosh, and R.~Ratasuk, ``An overview of {3GPP}
  device-to-device proximity services,'' \emph{IEEE Communications Magazine},
  vol.~52, no.~4, pp. 40--48, April 2014.

\bibitem{li2014exploring}
Y.~Li, L.~Sun, and W.~Wang, ``Exploring device-to-device communication for
  mobile cloud computing,'' in \emph{2014 IEEE International Conference on
  Communications (ICC)}.\hskip 1em plus 0.5em minus 0.4em\relax IEEE, 2014, pp.
  2239--2244.

\bibitem{li2015incentive}
P.~Li and S.~Guo, ``Incentive mechanisms for device-to-device communications,''
  \emph{IEEE Network}, vol.~29, no.~4, pp. 75--79, 2015.

\bibitem{tanbourgi2014cooperative}
R.~Tanbourgi, H.~Jakel, and F.~K. Jondral, ``Cooperative interference
  cancellation using device-to-device communications,'' \emph{IEEE
  Communications Magazine}, vol.~52, no.~6, pp. 118--124, 2014.

\bibitem{zhang2013interference}
R.~Zhang, X.~Cheng, L.~Yang, and B.~Jiao, ``Interference-aware graph based
  resource sharing for device-to-device communications underlaying cellular
  networks,'' in \emph{2013 IEEE Wireless Communications and Networking
  Conference (WCNC)}.\hskip 1em plus 0.5em minus 0.4em\relax IEEE, 2013, pp.
  140--145.

\bibitem{xu2013token}
J.~Xu and M.~Van Der~Schaar, ``Token system design for autonomic wireless relay
  networks,'' \emph{IEEE Transactions on Communications}, vol.~61, no.~7, pp.
  2924--2935, 2013.

\bibitem{mastronarde2016relay}
N.~Mastronarde, V.~Patel, J.~Xu, L.~Liu, and M.~van~der Schaar, ``To relay or
  not to relay: learning device-to-device relaying strategies in cellular
  networks,'' \emph{IEEE Transactions on Mobile Computing}, vol.~15, no.~6, pp.
  1569--1585, 2016.

\bibitem{zhang2015contract}
Y.~Zhang, L.~Song, W.~Saad, Z.~Dawy, and Z.~Han, ``Contract-based incentive
  mechanisms for device-to-device communications in cellular networks,''
  \emph{IEEE Journal on Selected Areas in Communications}, vol.~33, no.~10, pp.
  2144--2155, 2015.

\bibitem{kephart1991directed}
J.~O. Kephart and S.~R. White, ``Directed-graph epidemiological models of
  computer viruses,'' in \emph{Research in Security and Privacy, 1991.
  Proceedings., 1991 IEEE Computer Society Symposium on}.\hskip 1em plus 0.5em
  minus 0.4em\relax IEEE, 1991, pp. 343--359.

\bibitem{wang2003epidemic}
Y.~Wang, D.~Chakrabarti, C.~Wang, and C.~Faloutsos, ``Epidemic spreading in
  real networks: An eigenvalue viewpoint,'' in \emph{Reliable Distributed
  Systems, 2003. Proceedings. 22nd International Symposium on}.\hskip 1em plus
  0.5em minus 0.4em\relax IEEE, 2003, pp. 25--34.

\bibitem{bailey1975mathematical}
N.~T. Bailey \emph{et~al.}, \emph{The mathematical theory of infectious
  diseases and its applications}.\hskip 1em plus 0.5em minus 0.4em\relax
  Charles Griffin \& Company Ltd, 5a Crendon Street, High Wycombe, Bucks HP13
  6LE., 1975.

\bibitem{van2009virus}
P.~Van~Mieghem, J.~Omic, and R.~Kooij, ``Virus spread in networks,''
  \emph{Networking, IEEE/ACM Transactions on}, vol.~17, no.~1, pp. 1--14, 2009.

\bibitem{wifidirect}
{Wi-Fi Direct}, \url{http://www.wi-fi.org/discoverwi-fi/wi-fi-direct}.

\bibitem{ltedirect}
{LTE Direct}, \url{https://www.qualcomm.com/
  invention/research/projects/lte-direct}.

\bibitem{att}
{ATT Data Perks}, \url{https://www.att.com/att/dataperks/en/index.html}.

\bibitem{mailath2006repeated}
G.~J. Mailath and L.~Samuelson, \emph{Repeated games and reputations: long-run
  relationships}.\hskip 1em plus 0.5em minus 0.4em\relax Oxford university
  press, 2006.

\bibitem{doya2000reinforcement}
K.~Doya, ``Reinforcement learning in continuous time and space,'' \emph{Neural
  computation}, vol.~12, no.~1, pp. 219--245, 2000.

\end{thebibliography}
